\documentclass[letterpaper, 12pt]{article}

\usepackage[english]{babel}
\usepackage{hyperref}
\usepackage[utf8]{inputenc}
\usepackage[T1]{fontenc}
\usepackage{csquotes}

\usepackage[letterpaper,top=3cm,bottom=2cm,left=3cm,right=3cm,marginparwidth=1.75cm]{geometry}

\usepackage{amssymb,amsmath,amsthm}
\usepackage{graphicx}
\usepackage{setspace}

\usepackage[style=chicago-authordate,
            natbib=true, 
            backend=biber,
            uniquename=false,
            doi=false,isbn=false,url=false]{biblatex}
\usepackage[authordate]{biblatex-chicago}            
\AtEveryBibitem{
    \clearfield{urlyear}
    \clearfield{urlmonth}
}

\usepackage{amsthm}
\newtheorem{proposition}{Proposition}
\newtheorem*{definition*}{Definition}
\usepackage{chngcntr}
\counterwithout{figure}{section}
\counterwithout{table}{section}
\usepackage{varioref}
\usepackage[normalem]{ulem}
\usepackage{authblk}

\title{A Model of Enclosures: \\Coordination, Conflict, and Efficiency in \\ the Transformation of Land Property Rights\footnote{We thank participants from the 2024 conferences at ThReD (MIT) and SIOE (Chicago), and seminar or conference participants at Fordham, Cornell, the IAEE, The World Bank, and at Hunter College and The Graduate Center at CUNY. Klaus Deininger and Jonathan Morduch shared helpful suggestions on earlier versions. An online appendix with proofs and code to replicate and interact with all figures in the text can be found at  \url{https://jhconning.github.io/enclosure_book}.} \\ }
\author{Matthew J. Baker and Jonathan Conning \\ Hunter College and the Graduate Center \\
City University of New York}

\begin{document}

\maketitle  
\begin{abstract}
Economists, historians, and social scientists have long debated how open-access areas, frontier regions, and customary landholding regimes came to be enclosed or otherwise transformed into private property. This paper analyzes decentralized enclosure processes using the theory of aggregative games, examining how population density, enclosure costs, potential productivity gains, and the broader physical, institutional, and policy environment jointly determine the property regime. Changes to any of these factors can lead to smooth or abrupt changes in equilibria that can result in inefficiently high, inefficiently low, or efficient levels of enclosure and associated technological transformation. Inefficient outcomes generally fall short of second-best. While policies to strengthen customary governance or compensate displaced stakeholders can realign incentives, addressing one market failure while neglecting others can worsen outcomes. Our analysis provides a unified framework for evaluating mechanisms emphasized in Neoclassical, Neo-institutional, and Marxian interpretations of historical enclosure processes and contemporary land formalization policies.
\end{abstract}   

\newpage

%\tableofcontents

\newpage
\onehalfspacing

\section{Introduction}

The emergence of more secure and exclusive private property rights to land from more fluid possession-based systems is widely touted as one of history's most consequential institutional shifts. Leading thinkers, from Adam Smith and Karl Marx to Douglas North, have identified this transformation as central to the emergence of capitalism and modern economic growth \citep{smith1982, marx1992, north1973, brenner1976}. Yet substantial debate persists around the causes, efficiency, and distributional consequences of the enclosure processes that led to modern forms of private ownership.\footnote{The term `enclosure' has had debated meanings within economic history, at times encompassing not just changes in property but also in production organization or technology use. Our own focus aligns with understanding how these phenomena are entangled.  We will clarify its specific meaning within our model.} 

These debates remain critically relevant today, as an estimated 2.5 billion people reliant on land for their livelihoods face insecurity stemming from unclear statutory recognition or inadequate protection for their access and ownership claims \citep[][p. 7]{pearce2016}. UN Habitat estimates that less than ten percent of rural and urban land plots are formally titled in poorer developing countries \citep[][p. 26]{augustinus2003}. Although the potential benefits of secure, state-recognized ownership are widely acknowledged, fundamental questions remain about the timing, methods, and consequences of land formalization efforts. Societies must weigh the expected benefits against the costs of establishing and enforcing new property regimes, as well as the potential for displacing existing customary tenure arrangements \citep{deininger2003,aldenwily2018}. 

In this paper, we develop a flexible framework to analyze decentralized and potentially policy-influenced processes that result in more exclusive property claims to land. We model these processes as an aggregative game \citep{acemoglu2013} in which potential claimants weigh the benefits of securing land rents and the returns of potential technological improvements against the costs of establishing exclusive rights. As enclosure occurs, labor allocates across remaining unenclosed areas or by working in the enclosed sector. Our analysis explores how changing economic, political, and policy environments alter incentives, ultimately influencing the evolution of property regimes and their associated economic and social outcomes.

Our two-layered model of strategic action of enclosure and labor allocation distinguishes our approach from traditional analyses of open-access and customary property regimes. Canonical models like \citet{weitzman1974} and recent macro-development models \citep[e.g.][]{chen2023} measure misallocations attributed to free access or customary property regimes compared to a private property alternative that is assumed to be costless to impose and maintain. This makes it unsurprising that customary regimes, where maintaining possession requires effort, appear inefficient.  Our approach is to instead explicitly incorporate the costs of establishing and maintaining private property rights into the analysis and to endogenize the transition to private property via decentralized decisions. This allows us to analyze different environments where customary regimes might persist or coexist with more formal land ownership.  

Our framework is inspired by \citet{demeza1992} but we extend the framework to a more general class of aggregative games that allow us to examine how equilibrium outcomes respond to changes in a wide range of economic environments and policy interventions. Factors such as population density, potential productivity gains from enclosure, initial patterns of possession and resistance, and enclosure costs interact to shape land privatization outcomes. Changes in these factors can trigger gradual or sudden shifts in equilibria, with varying but determinate impacts on social efficiency and distribution. 

While the open-access case provides a useful analytical benchmark, it rarely reflects the complexities of real-world property regimes. As Elinor Ostrom and others have emphasized, communities typically regulate resource use through membership restrictions, social norms, and formal rules that mitigate the inefficiencies of pure open access \citep{ostrom1990,baland1996,bromley1992}. In addition, existing users often actively resist privatization and demand compensation for lost access rights. Our extended model incorporates these critical features - regulated access, political power dynamics, and compensation requirements - enabling analysis of diverse property regimes including customary tenure systems and contested claims.

We identify three key externalities that shape the efficiency of decentralized enclosure decisions. These arise endogenously from the interplay of individual enclosure decisions and the responses of other agents to shape whether these decisions lead to socially efficient outcomes. These externalities are the familiar rent-dissipating externality associated with unregulated access, positive spillovers from privatization-induced productivity improvements, and negative externalities that arise when enclosers capture rents from displaced users without compensation. The relative strength of these externalities determines the strategic nature of enclosure decisions. When positive spillovers dominate, enclosure decisions become strategic substitutes, dampening private incentives for socially efficient property rights transformation and technological change. Conversely, when negative externalities prevail, enclosure decisions become strategic complements, leading to the possibility of premature privatization and wasteful enforcement expenditure. We show that policies to improve governance and regulation of the commons and policies to compensate displaced users can realign incentives, but addressing one of these market failures without addressing the others can worsen equilibrium outcomes.

The remainder of this paper proceeds as follows. After reviewing related literature in Section \ref{relatedlit}, Section \ref{themodel} presents our baseline model and establishes efficiency benchmarks. Section \ref{efficiency} analyzes decentralized enclosure processes and the gap that may develop between private and social objectives. We establish a planner's second-best benchmark and explain why resource allocation under decentralized enclosure processes fails to achieve it. Section \ref{extended_model} incorporates regulated access and compensation requirements, better reflecting real-world property regime transitions. Section \ref{sec-extensions} explores applications ranging from historical analyses of English enclosure movements to frontier and colonial land policies, as well as modern land formalization programs in developing countries. Section \ref{conclusion} concludes.

\section{Related Literature} \label{relatedlit}
Neoclassical economics highlights the crucial role of secure and exclusive private property rights for efficient resource allocation, but largely ignores the practical costs and political complexities of establishing and enforcing such rights.\footnote{As Abba Lerner wryly observed, ``Economics has gained the title Queen of the Social Sciences by choosing solved political problems as its domain.'' \citep[p.259]{lerner1972}}  

Factors such as population density, commercialization, markets and technology, and state capacity affect these trade-offs, but interpretations differ on how changes to any of these factors might change individual incentives to induce institutional shifts in each of these theoretical approaches. Four key dimensions are useful for comparison: (1) whether change tends to be evolutionary or conflict-driven; (2) whether it tends to improve efficiency and raise productivity or whether inefficiency can arise and persist; (3) whether it leads to labor absorption or displacement and mitigates or increases inequality; and (4) whether change arises via decentralized processes, from below, or tends to be imposed from above.  One way of characterizing previous schools of thought on transformation of ownership is the emphasis and viewpoint - positive or negative - they have about each of these aspects of transformation. 

The New Institutional Economics literature, along with Austrian and Marxian perspectives, provides richer analytical narratives for analyzing the origin and transformation of land property rights, though with differing emphases and conclusions. Although a comprehensive review is beyond our scope,\footnote{An eclectic but very incomplete list of opinionated surveys and books of mostly economic writing on the topic might include \citet{binswanger1995, baland1996, besley2010, ostrom1990, libecap1993, anderson2004, bardhan1989, byres1996, murtazashvili2016a,desoto2000, otsuka2014}.} these traditions highlight that in many contexts the benefit of establishing and enforcing more exclusive property rights may entail costs that exceed the benefits.

Austrian economists \citep[e.g.][]{murtazashvili2016a} and the first New Institutionalists \citep{alchian1973,demsetz1967, north1973}, together with influential observers like \citet{boserup1965} and \citet{ostrom1990}, emphasize evolutionary adaptation to local circumstances, through individual entrepreneurship, community negotiation, and Coasian bargaining processes \citep{coase1960}. They argue that institutional adaptations emerge organically in response to changing economic conditions, tending toward efficiency, though this gradual evolution may be distorted and disrupted by external forces like state intervention or rapid technological change.

Later New Institutionalists including \citet{north1990}, \citet{binswanger1995}  and \citet{acemoglu2005}, place more emphasis on power relations, distributional conflicts, and coordination failures in shaping institutional change. They argue that transaction costs often stand in the way of institutional change and that conquerors and elites impose institutions to serve their interests, often leading to persistent societal inefficiency. This perspective provides a lens through which to examine traditional justifications for state intervention in land registration, titling, or land reform.  

Historically, interventions, particularly during colonial expansion and state formation, were often justified by claiming that customary or indigenous land tenure systems were inefficient. John Locke's theory of property, arguing that individuals gain land rights through labor and improvement, was often invoked to justify enclosure and colonization. Although Locke's arguments were more nuanced, others often stretched his theory, particularly his characterization of unimproved land as `waste', to provide justification for land takings and displacement of customary users \citep{greer2018, barbier2010}. Whether such interventions, often imposed from above, actually delivered their promised improvements in different contexts or primarily served as a pretext for expropriation by more powerful actors remains debated \citep{platteau1996,borrasjr2012,deininger2003}.

`Political Marxists' like Robert Brenner (\citeyear{brenner1976}) and Ellen Wood (\citeyear{wood2002}) offer a yet more conflict-driven perspective, emphasizing how shifts in the balance of class power served as a primary driver of institutional change. They argue that European customary property regimes, while granting peasants possession rights, were inherently inefficient and constrained by feudal legacies and class conflicts. This severely limited their ability to adapt and produce growth. They argue that the transition to capitalism required a shift in class power that enabled landlords to expropriate the peasantry via enclosures and other means. This created a landless, wage-dependent workforce and the conditions for the development of capitalist incentives or imperatives that led to continuous increases in productivity, labor release, and structural transformation.\footnote{A large ongoing literature has challenged the empirical and theoretical basis for several of these assertions \citep{aston1987,hatcher2001,whittle2013}} 

Although the perspective of the Political Marxists, with its focus on class conflict and expropriation, may seem distinct, \citet[][p.4]{allen1992} identifies a shared assumption he labels `agrarian fundamentalism': the belief that customary regimes are inefficient obstacles to progress. He provides evidence that, at least for the English South Midlands areas that he studied, significant growth in agricultural productivity had already occurred under the common field system and customary tenure \textit{prior} to the British Parliamentary enclosures. Similar revisionist evidence has been presented in other contexts \citep[see e.g.][]{clark1998, kopsidis2015, berry1993, platteau1996, ostrom1990}, all suggesting that customary tenure arrangements were often more efficient and productive than assumed and often had already evolved to allow for individualized plots, inheritance, transfers, and market operations. Allen suggests that the very success of these customary regimes and the increasing value of land made them targets for enclosure, driven by rent-seeking motives rather than solely by the pursuit of greater productivity.\footnote{Considerable debate on the issues persists. In sections \ref{who_encloses} and \ref{manuf_sec} we discuss how our model may help frame some aspects of the English enclosure debates.} 

Recent macro-development papers work within similar neoclassical models but extend their critique to customary property regimes, arguing that the misallocation that arises can account for low agricultural productivity, delayed structural transformation, and a large part of measured TFP differences across countries \citep[e.g.][]{chen2017a, chen2023, gottlieb2019}. However, they too assume that the transition to a more efficient private property alternative is exogenously driven and costless.  Their models cannot address the possibilities we consider: that customary land property regimes might be efficient in some circumstances; nor can they explain why they have persisted or coexisted, or where and why endogenous decentralized institutional transformation is likely to take place or not. They also typically assume lump-sum transfers to bypass the distributional issues that can arise when new private property ownership displaces customary users, where we will trace out distributional pathways.

Recent macro-development models offer valuable tools for empirically measuring misallocation by extending the neoclassical critique. However, by assuming costless and exogenously driven transitions to private property, these models implicitly embrace agrarian fundamentalism, as the customary property regime necessarily appears less efficient in all environments \citep{chen2017a,chen2023, gottlieb2019}.  In contrast, our framework explicitly models the costs and strategic interactions involved in property rights transitions, allowing us to analyze the persistence and potential efficiency of customary regimes in different environments, as well as the distributional consequences of enclosure.

While there is no substitute for deep historical and anthropological research to understand the nuances of property regimes and their transformation, formal models can help clarify the key mechanisms driving their evolution and predict their outcomes.\footnote{As \citet[][p. 1]{hatcher2001} notes, ``what historians write is imbued with an awareness of theory and abstract concepts,'' and that  their interpretations often ``resonate with the influence of grand but conflicting models of the processes of long-term change and development."}   This paper contributes to this discourse by developing a tractable model of decentralized enclosure processes.  By incorporating the costs of establishing and maintaining property rights and allowing for strategic interactions between different agents, our framework helps identify the circumstances under which various theoretical predictions are likely to hold.
 
\section{Benchmark model} \label{themodel}
\subsection{Technology and resources}

We begin with a model of a purely agricultural economy producing a single output, say, blueberries. Production uses land and labor, with land potentially existing in one of two states: enclosed or unenclosed. Enclosed land has an owner (or owners) who can exclude others at a cost. For now unenclosed land is in open access; anyone can gather blueberries from any parcel of this land that they can occupy by possession.  Section \ref{extended_model} extends the model to allow for community-based regulation of access to unenclosed, or `common', land.

Production in the unenclosed (or customary/common) and enclosed sectors follows Cobb-Douglas technologies:
\begin{align*} 
\text{Unenclosed:}  \quad  A F(T_{c},L_{c}) &= A T_{c}^{1-\alpha} L_{c}^{\alpha} \\  
\text{Enclosed:}  \quad  \theta A F(T_{e},L_{e}) &= \theta A T_{e}^{1-\alpha} L_{e}^{\alpha}
\end{align*}\label{private_prod_fun}
where $T_{i}$ and $L_{i}$ denote land and labor in sector $i \in \{c,e\}$, $A$ is total factor productivity, and $\theta$ captures potential productivity differences between sectors. While we adopt Cobb-Douglas functions for tractability, our core comparative static results generalize to any concave, monotonic technology due to the properties of aggregative games.\footnote{Most results require only that players' payoff functions are continuous and concave in the overall enclosure rate \citep{acemoglu2013}.}

The parameter $\theta$ captures potential productivity gains from enclosure. When $\theta > 1$, enclosure enables technological or organizational improvements - for example, fenced plots may protect blueberries from trampling or animal foraging, or more secure property rights may incentivize investment in higher-yield varieties \citep{boserup1965}. When $\theta \leq 1$, enclosure offers no productivity advantage or may even reduce output, as when fencing reduces cultivable area. By modeling post-enclosure changes as homothetic transformations, we isolate how factors beyond technology drive the adoption of different production techniques on enclosed land.

The economy has fixed endowments of land ($\bar T$) and labor ($\bar L$), with population density defined as $\bar l = \frac{\bar L}{\bar T}$. Of these totals, if $T_e$ units of land and $L_e$ units of labor are employed in the enclosed sector, there will be $T_c=\bar T-T_e$ and $L_c=\bar L-L_e$ in the unenclosed sector. We denote the shares of enclosed land and labor as $t_e=\frac{T_e}{\bar T}$ and $l_e=\frac{L_e}{\bar L}$ respectively.

In the enclosed sector, competitive firms employ factors under constant returns to scale, paying market wages and land rents (in units of output). Labor freely moves between sectors until returns equalize: Workers can either earn market wages in the enclosed sector or directly produce output in the unenclosed sector. 

Before analyzing the decentralized equilibrium, in the next section, we first characterize the social planner's problem to determine when enclosure and associated technological improvements are optimal. This benchmark will allow us to evaluate the efficiency of outcomes that emerge from the strategic interaction between decentralized enclosers in subsequent sections. It will be useful to collect the key parameters describing the economic environment in the vector $\varphi =(A, \theta, \bar l, \alpha, c)$, where $A$ captures baseline productivity, $\theta$ measures potential gains from enclosure, $\bar l$ represents population density, $\alpha$ is the labor share in production, and $c$ denotes enclosure costs.

\subsection{First-best labor allocation and enclosure}\label{planner}

A social planner chooses both the share of land to enclose ($t_e$) and the share of labor to allocate to the enclosed sector ($l_e$) to maximize total output net of enclosure costs. With total land $\bar T$ and labor $\bar L$, this implies solving:
\begin{equation}\label{fbmaxxer}
\max_{t_e, l_e} \quad A  \left [ \theta F(t_e \bar T, l_e \bar L) + F((1-t_e)\bar T, (1-l_e)\bar L) \right ] - c t_e \bar T
\end{equation}
where $c$ represents the per-unit cost of enclosure.\footnote{Following \citet{demeza1992}, we assume linear enclosure costs: at cost $c$ per unit of land, an owner can establish and enforce exclusive property rights. For now, we also follow the canonical literature \citep{weitzman1974,samuelson1974,cohen1975} in assuming enclosers are outsiders and existing users are displaced without compensation.} Under Cobb-Douglas production,  we can write $\theta A F(T_e, L_e)=\theta A F(t_e,l_e)\cdot F(\bar T, \bar L)$ and express this objective more compactly as:

\begin{equation} \label{fbmaxxer2}
\max_{t_e, l_e} \quad \left [ \theta F(t_e, l_e) + F(1-t_e, 1-l_e) \right ] \cdot A \bar l^\alpha - c \cdot t_e
\end{equation}
where $A \bar l^\alpha = {AF(\bar T, \bar L)}/{\bar T}$ is potential output per unit land using the non-augmented technology. A necessary condition for efficient labor allocation is that the marginal product of labor be equalized across plots and sectors, as can be seen from differentiation of (\ref{fbmaxxer2}) with respect to $l_e$. We can express the condition $MP_L^e=MP_L^c$ in terms of land and labor shares allocated to the enclosed sector as follows: 

\begin{align} \label{opt_labor} 
    \theta A \cdot F_L(T_e,L_e) &= A \cdot F_L(\bar T-T_e,\bar L-L_e)
\end{align}
For our parameterization, (\ref{opt_labor}) reduces to:

\begin{equation} \label{mplfoc}
  \theta \alpha A \left ( \frac{t_e}{l_e} \right ) ^{1-\alpha} = \alpha A \left ( \frac{1-t_e}{1-l_e} \right ) ^{1-\alpha}
\end{equation}
This equation can be solved to find the \textit{planner's first-best labor allocation function}. This tells us, for any initial share of enclosed land $t_e$ (which may not be the efficient choice), the share of labor $l_e$ the planner would allocate to the enclosed sector to maximize total output:

\begin{equation} \label{fblab}
l_e^1(t_e) = \frac{\Lambda_1 t_e}{1+(\Lambda_1-1)t_e}, \quad \text{where }\Lambda_1=\theta^{\frac{1}{1-\alpha}}.
\end{equation}
 
This labor share function depends only on the relative productivity parameter $\theta$ and not on the total factor productivity level $A$. When $\theta>1$, we must have $\Lambda_1>1$ in (\ref{fblab}) which implies that $l_e^1(t_e) \ge t_e$ for all $t_e \in [0,1]$ (with strict inequality for interior points). Thus, at any interior enclosure rate $t_e$, the planner employs \textit{higher} (or never lower) labor intensity on enclosed land compared to unenclosed land when $\theta>1$, so $\frac{L_e}{T_e}\ge\frac{\bar L}{\bar T}\ge\frac{L_c}{T_c}$. This labor \textit{augmentation} contrasts with the labor allocation reaction functions under decentralized enclosure situations studied below, where private enclosure leads to labor \textit{displacement} on newly enclosed plots unless the expected productivity gain $\theta$ is above a high threshold $\theta_H$. 

To facilitate comparison with decentralized outcomes, we reformulate the planner's maximization problem in two steps. Working backwards, the planner first determines the optimal labor intensity function $l_e^1(t_e)$ for any enclosure rate $t_{e}$. This optimal labor share function (given by (\ref{fblab})) is then incorporated into the objective to reformulate the problem in terms of the choice of the enclosure rate $t_{e}$, which solves:

\begin{equation}\label{zte}
\max_{t_e} \quad z_1(t_e) - c \cdot t_e
\end{equation}
where
\begin{equation*}
z_1(t_e) = \left [ \theta F(t_e, l_e^1(t_e)) + F(1-t_e, 1-l_e^1(t_e)) \right ] \cdot A \bar l^\alpha 
\end{equation*}
which, using the functional form in (\ref{fblab}) can be simplified to:
\begin{equation*} 
\ z_1(t_e) = \left[ 1+\left(\Lambda_1-1\right)t_e\right]^{1-\alpha} \cdot A \bar l^{\alpha} 
\end{equation*}

Enclosure is costly, so it will only be socially optimal when it enables higher productivity technology (that is, where $\theta>1$) and effective population pressure $A \bar l^\alpha$ is sufficiently high relative to enclosure costs $ct_e$. Since, by assumption, the planner ensures efficient labor allocation across sectors (the planner perfectly regulates access to the not yet enclosed or improved areas), the enclosure decision reduces to a straightforward technology adoption decision: enclose land if and only if the output gains from technological improvement $\theta$ and associated optimal labor reallocation exceed the marginal cost $c$. This leads to our first result:

\begin{proposition}[First-Best Enclosure]\label{firstbest} When $\Lambda_1>1$, the planner chooses enclosure rate $t_e^1$ as follows:
\begin{itemize}
    \item  \textbf{No Enclosure:} $t_e^1=0$, when $z_1'(0)\leq c$.   
    \item \textbf{Full Enclosure:} $t^1_{e}=1$, when $z_1'(1) \geq c$.
    \item \textbf{Partial Enclosure:} $t_e^1 \in (0, 1 )$, when $z_1'(0)>c,$ $z_1'(1)<c$, and $z_1'(t_e^1)=c$.
\end{itemize} 
When $\Lambda_1<1$, the planner always chooses no enclosure.

\end{proposition} 
\begin{proof}[Proof:]
For $\alpha \in (0,1)$ and $\theta>1$ (equivalently  $\Lambda_1>1$), the planner's benefit function $z_1(t_e)$ in equation (\ref{zte}) is strictly concave ($z_1''(t_{e})<0$). Therefore, to determine whether land should remain unenclosed, partially enclosed, or fully enclosed, we only need to compare the marginal benefit of enclosure $z'_1(t_e)$ with its marginal cost $c$ at the endpoints, which determines whether $z_1(t_e)-ct_e$ is maximized at a corner or interior solution. When $\Lambda_1<1$, $z_1(t_e)$ is strictly decreasing in $t_e$, making zero enclosure optimal.
\end{proof}

Using Proposition \ref{firstbest}, we can map our model parameters $\varphi$ to optimal enclosure decisions. Differentiating $z_1(t_{e})-ct_e$ with respect to $t_e$ and rearranging defines a parametric region coinciding with the case in which `no enclosure' is socially optimal outlined by Proposition \ref{firstbest}. We can express this as a critical threshold for population density $\bar l$:

\begin{equation} \label{sw_noenc_rf}
z_1'(0) \leq c \quad \Leftrightarrow \quad    \bar l \le 
\left[
   \frac{1}{(1-\alpha)\left(\Lambda_1-1\right)} \cdot \frac{c}{A}
\right]^{\frac{1}{\alpha}}=\bar l_0^1(\theta)
\end{equation}

Similarly, the `full enclosure' condition requires $z_1'(1)\geq c$. which yields:

\begin{equation} \label{sw_fenc_rf}
z_1'(1)\geq c \quad \Leftrightarrow \quad \bar l \geq \Lambda_1 \cdot 
\bar l_0^1(\theta)
=\bar l^1_1(\theta)
\end{equation}

Partial enclosure is optimal for population densities between the no-enclosure threshold (\ref{sw_noenc_rf}) and the full enclosure threshold (\ref{sw_fenc_rf}).

\begin{figure}[htb!] 
\centering
\includegraphics[width=.9\textwidth]{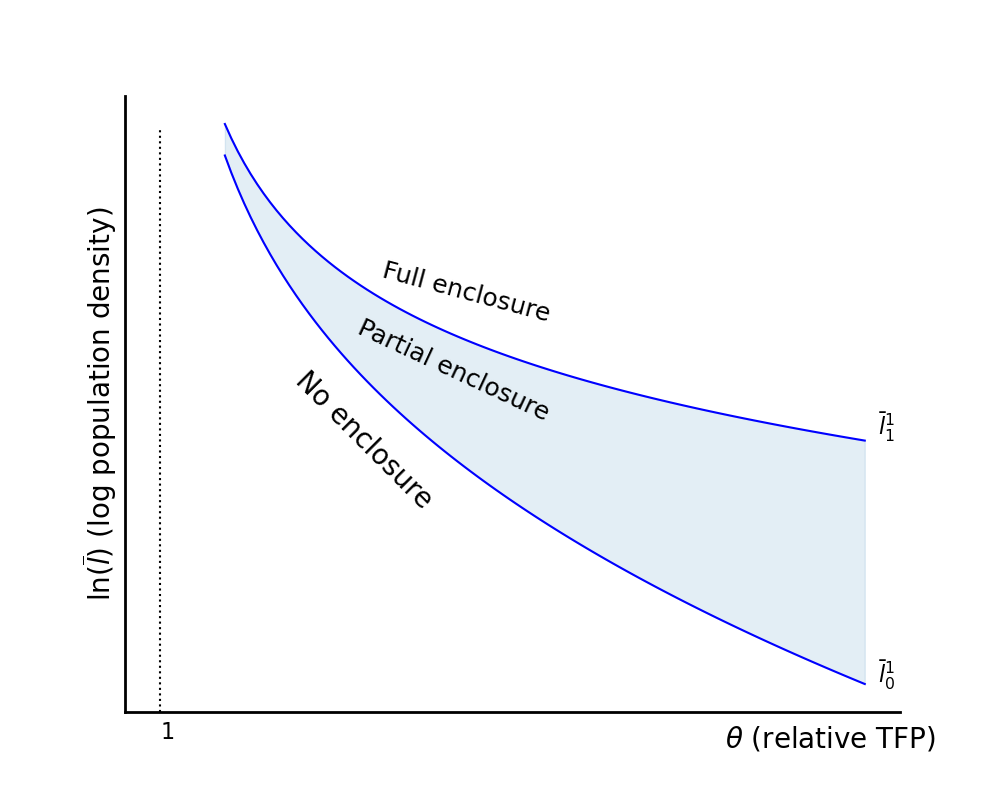}
\caption{Socially efficient enclosure decisions as they depend on (log) population density $\bar l$ and expected TFP gain from enclosure $\theta$.}
\label{fig-social}
\end{figure}  

Figure \ref{fig-social} is a graphical representation of the parametric description of Proposition \ref{firstbest}, and shows the loci $\bar l_0^1$ and $\bar l_1^1$ plotted as solid lines in ($\theta$, $\ln \bar l$) space.\footnote{We plot the natural logs of the loci, for better scaling.  This and all subsequent diagrams are drawn for parameters $c/A=1$ and $\alpha=2/3$.} The shaded parameter region between the two loci corresponds to situations where partial land enclosure is optimal. 

This analysis suggests a Boserupian interpretation of induced agricultural intensification: when land is abundant (low $\bar l$), enclosure benefits fall short of costs unless potential productivity gains $\theta$ are sufficiently high. Thus, while $\theta>1$ is necessary for enclosure to be optimal when $c>0$, it is not sufficient. Even with productivity-enhancing enclosure ($\theta>1$), the opportunity cost of drawing labor from the unenclosed sector plus enclosure costs may outweigh potential gains when land is abundant (low $\bar l$) or baseline productivity is high (high $A$).

Consider an economy initially below the $\bar l_0^1$ threshold where no enclosure is optimal. If population growth pushes density $\bar l$ above this threshold, the marginal social return to enclosure now exceeds its cost, leading the planner to enclose some land and operate it at higher labor intensity. However, since drawing labor from the unenclosed sector incurs increasing marginal costs, the planner at first encloses only enough land to efficiently accommodate the population increase.

\subsection{Decentralized enclosure processes}\label{market}

We now analyze how private actors' decentalized strategic interactions determine enclosure outcomes. Each unit of land has a potential encloser who can, at cost $c$, establish exclusive rights. Following \citet{weitzman1974} and \citet{demeza1992}, we initially assume enclosers do not compensate displaced users, though Section \ref{power} extends the model to settings where enclosers must pay compensation or overcome resistance.\footnote{More elaborate contest functions and strategic aspects of enforcing claims can be readily incorporated. See \citet{baker2003, baker2008, hafer2006}, and the discussion in Section \ref{policy_regulate}.}

Upon enclosure, each unit of land enters a competitive formal sector employing technology $\theta F(T,L)$. This sector pays competitive factor returns: land earns a rental rate $r$ equal to its marginal product $MP_T^e=\theta AF_T$, while labor earns a wage $w=MP_L^e$.\footnote{This formulation is equivalent to enclosers either directly hiring labor and earning profits or charging competitive access tolls to laborers (see \citet[][p. 230]{weitzman1974}).} 

Labor chooses between this market wage and establishing possessory rights on open access land. With symmetric ability to contest possession, the $L_c$ units of labor spread evenly across $T_c$ units of unenclosed land. Each worker thus operates $T_c/L_c$ units of land and earns the average product $AP_L^c = F(T_c, L_c)/L_c$ which, by Euler's theorem, decomposes into labor's marginal product plus implied rents from possessed land:

\begin{equation} \label{commondecomp}
AP_L^c = MP_L^c + MP_T^c \cdot \frac{T_c}{L_c}
\end{equation}
In equilibrium, labor must be indifferent between working in either sector:  

\begin{equation} \label{labmkteq2}
MP_L^c < MP_L^c + MP_T^c \cdot \frac{T_c}{L_c}  = w_e = MP_L^e  
\end{equation}
The lower marginal product of labor in the unenclosed sector, $MP_L^c<MP_L^e$, demonstrates the classic result of higher-than-efficient labor intensity, or `overcrowding,' on open access land also known as the `tragedy of the commons' \citep{dasgupta1979, richardcornes1996, baland1996}.
 
For our Cobb-Douglas specification, the equilibrium condition $AP_L^c=MP_L^e$ becomes:

\begin{equation} \label{MPLS}
    \theta \alpha A \left(\frac{t_e}{l_e}\right)^{1-\alpha}
    = A \left(\frac{1-t_e}{1-l_e}\right)^{1-\alpha}
\end{equation}
This equilibrium condition yields the \textit{labor reaction function} which determines the share of labor $l_e$ allocated to the enclosed sector for any enclosure rate $t_{e}$: 

\begin{equation} \label{optle0}
    l_e^0(t_e) = \frac{\Lambda_0 t_e}{1+(\Lambda_0-1) t_e }  \quad \text{where } \Lambda_0=(\alpha\theta)^\frac{1}{1-\alpha}
\end{equation}
The parameter \(\Lambda_0\) in the labor reaction function (\ref{optle0}) differs from $\Lambda_1$ in the planner's labor allocation function (\ref{fblab}), with $\Lambda_0 \le \Lambda_1$. Consequently, at any interior enclosure rate \(t_{e} \in (0,1)\), the decentralized equilibrium allocates strictly less labor to enclosed land than the planner would choose, so $l_e^0(t_e)<l_e^1(t_e)$. 

Unlike the planner's optimal labor allocation function (\ref{fblab}), which is always concave, the labor reaction function (\ref{optle0}) can be convex or concave. Specifically, when $\Lambda_0<1$, $l_e^0(t_e)$ is convex to the origin; otherwise it is concave. This qualitative difference proves critical for the model's behavior and can be characterized by a threshold value $\theta_H$:

\begin{definition*}[High-TFP Threshold] 
Let $\theta_{H}=\frac{1}{\alpha}$. Then:

    When $\theta\geq \theta_{H}$: \textbf{enclosure introduces high-TFP plot level gains} 

    When $\theta < \theta_{H}$: \textbf{enclosure introduces low-TFP plot level gains}
\end{definition*}
\vspace{.35cm}

Note that the low-TFP case includes situations where $\theta<1$, implying productivity losses from enclosure.

The private return to enclosure equals the marginal product of land, which depends on the labor reaction function:
\begin{align} \label{eqr_rf}
    r(t_e) &= {\theta} \cdot A  F_T(t_e \bar T, l_e^0(t_e) \bar L) 
\end{align}
Substituting the labor reaction function (\ref{optle0}) into the parametric form of (\ref{eqr_rf}) yields:

\begin{equation} \label{eq_rf_cf} 
    r(t_e)={\theta}\cdot(1-\alpha)\cdot  A \bar l^{\alpha} \cdot 
    \left ( \frac{\Lambda_0}{1+(\Lambda_0 -1)t_e} \right)^{\alpha}
\end{equation}
Atomistic potential enclosers compare this rental rate to the enclosure cost $c$, choosing to enclose when $r(t_e) > c$. In a Nash equilibrium, each agent makes a best-response enclosure decision taking others' enclosure actions (as well as labor reactions) as given. Since these actions affect payoffs only through the aggregate enclosure rate $t_e$, this process has the structure of an aggregative game \citep[see][]{corchon2021,acemoglu2013}. The nature of this game changes fundamentally depending on whether enclosure introduces high or low-TFP gains.

\begin{proposition}[Strategic Interactions in Enclosure Decisions]\label{compssubs} 
Let $\theta_{H}=\frac{1}{\alpha}$. When:

$\theta \geq \theta_H$ enclosure decisions are strategic substitutes; $r'(t_e) \leq 0$. 

$\theta < \theta_H$ enclosure decisions are strategic complements; $r'(t_e)>0$. 
\end{proposition} 
\begin{proof}[Proof:]
The rental rate function $r(t_{e})$ in (\ref{eq_rf_cf}) inherits its properties from the labor reaction function (\ref{optle0}).  We have $r'(t_e)<0$ when $\Lambda_0>1$ which requires $(\alpha\theta)^\frac{1}{1-\alpha}>1$, or equivalently, $\theta>\frac{1}{\alpha}$, coinciding with the definition of high-TFP gains to enclosure. Similarly, $r'(t_e)>0$ when $\Lambda_0<1$, which occurs when $\theta<\frac{1}{\alpha}$, in the low-TFP gains to enclosure region.
\end{proof}

This result reveals a fundamental dichotomy in decentralized enclosure processes. In high-TFP economies ($\theta>\theta_H$), enclosure decisions are strategic substitutes as the marginal return $r(t_e)$ decreases with aggregate enclosure. This leads to equilibrium behavior characterized in:

\begin{proposition}[Decentralized Enclosure with High-TFP gains]\label{hightfpenc} When $\theta>\theta_H$, pure strategy equilibria are characterized by: 
\begin{itemize}
    \item  \textbf{No enclosure:} $t_e^*=0$, when $r'(0)<c$.   
    \item \textbf{Full enclosure:} $t^*_{e}=1$, when $r'(1) \geq c$.
    \item \textbf{Partial enclosure:} $t_e^* \in (0, 1 )$, when $r'(0)>c,$ $r'(1)<c$, and $r'(t_e^{o})=c$.
\end{itemize} 
\end{proposition} 
\begin{proof}[Proof:]
Whenever $\Lambda_0>1$, differentiation of (\ref{optle0}) shows $r'(t_e)<0$, and $r''(t_e)>0$. The results follow from comparing $r(t_e)$ with $c$ at the boundaries and interior points.
\end{proof}

While this high-TFP case yields equilibria qualitatively similar to the first-best outcomes in Proposition \ref{firstbest}, the low-TFP case ($\theta < \theta_H$) produces starkly different results. Here, enclosure decisions become strategic complements as $r(t_e)$ increases with aggregate enclosure. As noted by \citet{vives2005} and \citet{acemoglu2013}, aggregative games with strategic complementarity can lead to multiple equilibria and coordination failures, as we show in the next Proposition. 

\begin{proposition}[Decentralized Enclosure with Low-TFP gains]\label{lowtfpenc} In the low-TFP gains region, equilibria are characterized by: 
\begin{itemize}
    \item  \textbf{No enclosure} is the unique dominant strategy equilibrium when $r(1)<c$.   
    \item \textbf{Full enclosure} is the unique dominant strategy equilibrium when $r(0)>c$.
    \item \textbf{Multiplicity} (either full or no enclosure) exist when $r(0)<c$ and $r(1)>c$.
\end{itemize} 
\end{proposition} 
\begin{proof}
When there are low-TFP gains, $r(t_e)$ is strictly increasing in $t_e$. Thus, if $r(1)<c$, no encloser has an incentive to enclose even under the best circumstances, as $r(t_e)<r(1)<c$ for all $t_e$. Conversely, if $r(0)>c$, enclosure always pays off as $r(t_e)>r(0)>c$ for all $t_e$. When $r(0)<c<r(1)$, either equilibrium is possible.
\end{proof}

The contrast between Propositions \ref{hightfpenc} and \ref{lowtfpenc} reflects fundamentally different labor market dynamics. In low-TFP economies, enclosed land employs less labor than the same plot when unenclosed. This labor displacement lowers equilibrium wages, raising labor utilization and rental rates across all land, and the return to further enclosures. Conversely, in high-TFP economies, each enclosure raises labor demand enough to increase plot-level employment, drawing in labor and reducing crowding on unenclosed land. While this raises wages, it reduces returns to further enclosure.

We can characterize these equilibria in terms of population density thresholds using the two propositions \ref{hightfpenc} and \ref{lowtfpenc}. Substituting $t_e=0$ into (\ref{eqr_rf}) we find an expression for the threshold level of population density $\bar l$ above which enclosure becomes attractive ($r(0)>c$):
\begin{equation} \label{noenc_dec_rf} 
r(0) \geq c \quad \Leftrightarrow \quad \overline{l} \geq \frac{1}{\Lambda_0}
\left[\frac{c}{\theta A} \cdot
\frac{1}{(1-\alpha)}\right]^\frac{1}{\alpha}=\bar l_0^d
\end{equation} 

Similarly, the region where $r(1)\geq c$ holds defines a critical region where full enclosure will occur:

\begin{equation} \label{fullenc_dec_rf}
r(1) \geq c \quad \Leftrightarrow \quad \overline{l} \geq 
\left[\frac{c}{\theta A} \cdot \frac{1}{(1-\alpha)}\right]^\frac{1}{\alpha}=\bar l_1^d
\end{equation}

The relative position of these thresholds depends on whether TFP gains are high or low. Under low-TFP gains, strategic complementarity implies the $r(1)>c$ threshold lies above the $r(0)>c$ threshold, creating the possibility of multiple equilibria. These thresholds are plotted in Figure \ref{figure_private}, intersecting at $\theta_H=\frac{1}{\alpha}$ which separates the high and low-TFP regions.

\begin{figure}[htb!] 
\centering
\includegraphics[width=.9\textwidth]{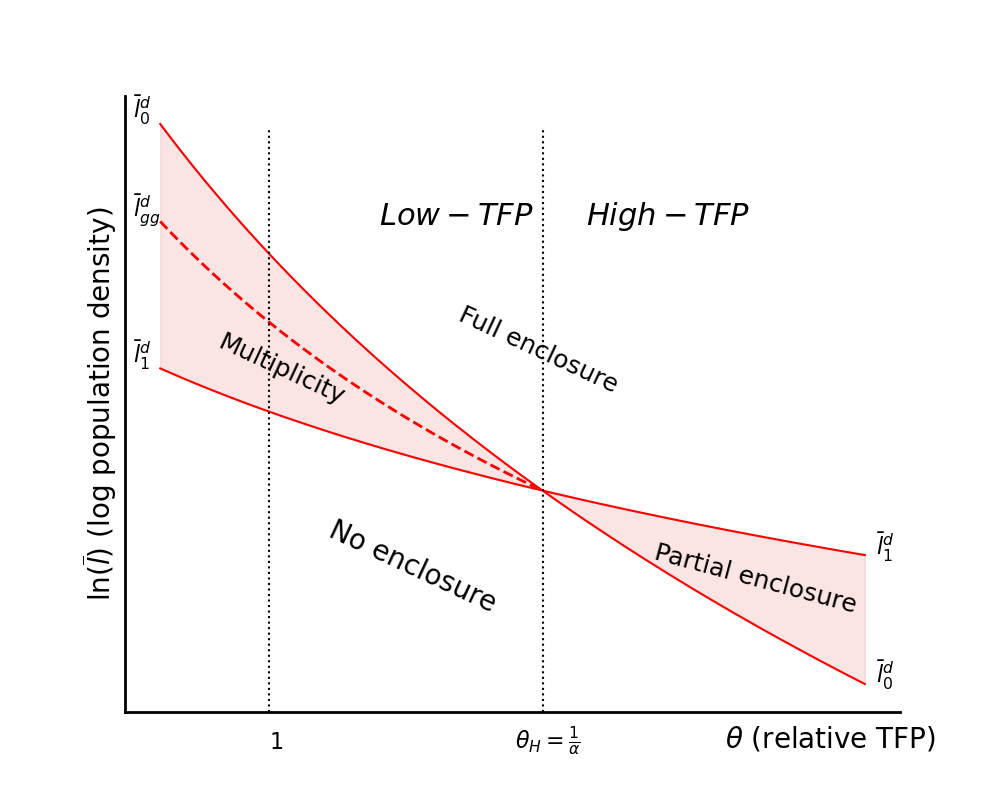}
\caption{Decentralized enclosure equilibria.}
\label{figure_private}
\end{figure}

In the high-TFP region between these thresholds, partial enclosure emerges as a stable equilibrium. However, in the low-TFP region, both `no enclosure' and `complete enclosure' can be Nash equilibria,\footnote{An unstable interior equilibrium with partial enclosure also exists.} as detailed in Proposition \ref{lowtfpenc}.  When no land is enclosed so $t_e=0$, it is privately unprofitable for any one claimant to enclose land, yet each would find it profitable to enclose if all others did as well. In such environments, small changes in beliefs about others' enclosure decisions could trigger rapid transitions between property regimes, leading to sudden changes in output and distribution. In this region, one might wish to deploy some sort of equilibrium refinement, as we now suggest.

\subsection{Multiple equilibria and global games}\label{section_globalgames}

The theory of Global Games \citep{carlsson1993, morris2003} providesd a natural framework for selecting between full and no enclosure equilibria in the multiple equilibria, low-TFP region, It is a convenient means for identifying risk-dominant enclosure equilibria. Consider introducing small uncertainty about the TFP parameter $\theta$. Potential enclosers observe a private signal  $x=\theta+\sigma \varepsilon_i$, where $\varepsilon$ has mean zero and finite support within the low-TFP region. As $\sigma$ approaches zero, parametric uncertainty about $\theta$ vanishes but strategic uncertainty remains.  As 
\citet{morris2003} shows, this leads to a unique cutoff strategy in Nash equilibrium: players enclose if and only if $\theta$ exceeds some threshold $\theta^*$, resulting in full enclosure whenever $\theta \geq \theta^*$.\footnote{Following \citet{morris2003} this  requires regions where each strategy is dominant. Proposition \ref{lowtfpenc} ensures these dominance criteria exist for certain values of $\theta$.} 

The cutoff $\theta^*$ is determined by the $\theta$ at which a potential claimant is indifferent between enclosing or not, given a uniform expectation about the fraction of others enclosing (i.e., a Laplacian belief about the enclosure decisions of others). Specifically,  $\theta^*$ (embedded in $\Lambda_0=(\alpha\theta^*)^\frac{1}{1-\alpha}$) solves:

\begin{equation*} \label{globalgameeqcond} 
    \int_0^1 \left[\Lambda_0\frac{(1-\alpha)}{\alpha} \cdot A \bar l^{\alpha} \cdot (1+\Lambda_0 t_e-t_e)^{-\alpha}-c\right]dt_e = 0
\end{equation*}

This condition can be simplified to show that full enclosure is the risk-dominant equilibrium when:

\begin{equation} \label{ggcond}
E[r(t_e)-c] = 0 \quad \Leftrightarrow \quad   \bar l \geq 
    \left[\frac{c}{\theta A} \cdot \frac{1-\Lambda_0}{\Lambda_0^\alpha - \Lambda_0}\right]^\frac{1}{\alpha}=\bar l_{gg}^d
\end{equation}
The locus is plotted as a dashed red line $\left ( \bar l_{gg}^d \right )$  in Figure \ref{figure_private}. In the low-TFP region, full enclosure emerges above this threshold, while no enclosure prevails below it. 

\subsection{Comparative statics}

The analysis and figures above help to describe how the equilibrium mix of property regimes -- the amount of land allocated to open-access or enclosed plots -- will change with key environmental parameters, including population density $\bar l$ and the potential for technological improvement after enclosure $\theta$. These variables appear on the vertical and horizontal axes of Figure \ref{fig-social} and all subsequent figures.  

The base-level total factor productivity parameter $A$ and the cost per unit land of enclosure $c$ also matter for enclosure decisions. These parameters enter as $c/A$ in all the boundary loci of enclosure decisions (expressions \ref{sw_noenc_rf}, \ref{sw_fenc_rf}, \ref{noenc_dec_rf}, \ref{fullenc_dec_rf}, and \ref{ggcond} above, and similar expressions below). An increase in $c$ generates an upward vertical displacement of all loci, making enclosure less likely, all else equal. Similarly, an increase in the base-level total factor productivity $A$ generates a downward vertical displacement, making enclosure more likely. We analyze the effect of additional policy and institutional variables in Section \ref{extended_model}.

\section{The Social efficiency of private enclosure decisions}\label{efficiency}

Having characterized both optimal and decentralized outcomes, we can now examine when decentralized enclosure processes lead to socially efficient labor allocations, why they sometimes fail, and their consequences for wages, rents, and distribution. This section analyzes these questions and establishes a second-best benchmark to guide policy interventions.

Figure \ref{fig_compare} superimposes the private enclosure decision regions of Figure \ref{figure_private} on the social planner decision regions of Figure \ref{fig-social}. This reveals where private and social enclosure decisions diverge and hence where decentralized equilibrium allocations $(t_e^*, l_e^0(t_e^*))$ fall short of potential output. In several parameter regions, decentralized enclosure produces efficient outcomes. For economies with low TFP gains and low population density (lower left region of Figure \ref{fig_compare}), the absence of private enclosure matches the social planner's choice. Similarly, for economies with high population density and high TFP gains (upper right region), decentralized processes efficiently achieve full enclosure.

\begin{figure}[htb!] 
\centering
\includegraphics[width=.9\textwidth]{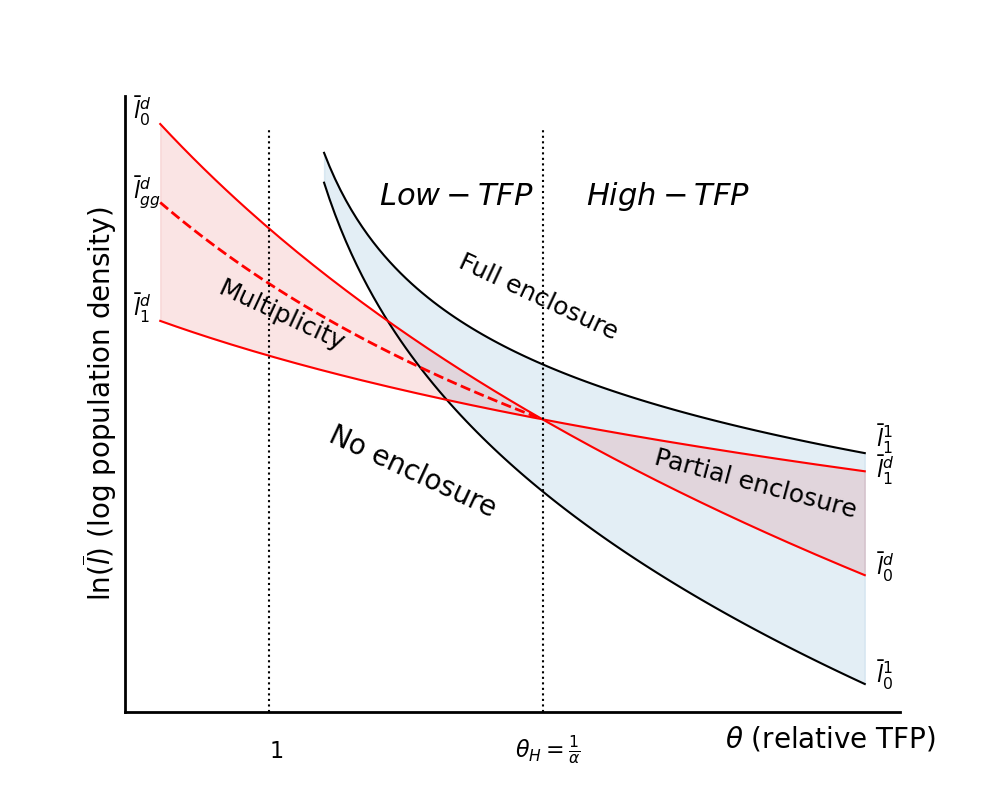}
\caption{Optimal and Decentralized Enclosures Compared. Red-shaded areas describe regions with partial private enclosure or multiple equilibria. Blue-shaded regions describe socially optimal partial enclosure.}
\label{fig_compare}
\end{figure}

Private and social equilibria can diverge in two ways. First, decentralized processes may lead to excessive enclosure. Using Global Games equilibrium selection, this occurs in the parameter region above the dashed $\bar l_{gg}^d$ locus and below the planner's $\bar l_0^1$ locus. Here, high population density drives privately captured rents above enclosure costs, even when net social benefits are negative. This inefficiency arises because (in the absence of compensation policies) private enclosers do not internalize the value lost by those they displace. Strategic complementarities imply that as some begin to enclose, the return to further enclosure rises. Alarmingly, this may precipitate sudden cascades or property races toward inefficient full enclosure even in economies where $\theta<1$ and enclosure leads to actual technical regress.

The decentralized economy may also under-provide enclosure relative to the social optimum. This occurs in economies to the right of $\theta_H$, above locus $\bar l_0^1$ but below locus $\bar l_0^d$. Despite high potential productivity gains, enclosers cannot appropriate enough of these gains to choose the socially efficient level of enclosure. In Section (\ref{prelim}) we identify the precise sources of these market failures.

\subsection{Are decentralized enclosures second-best?} \label{secondbest}

While we can identify where private and social enclosure rates diverge, evaluating these divergences requires a more nuanced policy benchmark than the first-best optimum. The decentralized economy faces labor misallocation problems that the first-best planner simply assumes away. A more policy-relevant question is whether a planner—who is also unable to regulate labor entry into unenclosed areas—could achieve higher net output by choosing different enclosure rates. 

Governments often control which land areas can be enclosed through settlement or zoning laws that delineate where land claims will be recognized and registered. However, these same governments may be unable to prohibit squatting or regulate access to unenclosed frontier regions and informal settlement areas.\footnote{U.S. land policy from colonial times through the nineteenth century repeatedly attempted to establish and maintain frontier lines, while struggling to prevent squatters from settling beyond official borders \citep[see e.g.][]{murtazashvili2013, alston2012}. Similar challenges with informal settlements and squatting have characterized urban development efforts from Lima to Dar es Salaam \citep[][]{desoto2000}.}

To establish the constrained (second-best) optimum, let $z_0(t_e)$ describe output per unit land achievable at different enclosure rates $t_e$, subject to the constraint that labor can still move freely into any remaining unenclosed land until equilibrium condition (\ref{commondecomp}) or (\ref{MPLS}) holds. We incorporate this constraint by substituting the labor reaction function $l^0(t_e)$ from (\ref{optle0}) into the planner's objective (\ref{fbmaxxer}):

\begin{equation}\label{zte0}
z_0(t_e) = \left [ \theta F(t_e, l_e^0(t_e)) + F(1-t_e, 1-l_e^0(t_e)) \right ] \cdot A \bar l^\alpha 
\end{equation}
After some simplification, the constrained planner's overall objective can be written using our parametric form as:
\begin{equation} \label{condmaxxer}
\begin{split}
z_{0}(t_e)-c \cdot t_{e} = \frac{\theta \Lambda_0^\alpha t_e+(1-t_e)}{(1+(\Lambda_0 - 1)t_e)^\alpha}
\cdot A \bar l ^{\alpha}
-c \cdot t_e
\end{split}
\end{equation}
This differs qualitatively from the unconstrained output function $z_1(t_e)$ in (\ref{zte}). While $z_0(t_e)=z_1(t_e)$ at the endpoints $t_e=0$ and $t_e=1$ where no labor misallocation occurs, $z_0(t_e)$ lies below $z_1(t_e)$ at all interior points dsince labor is misallocated relative to the first-best optimum.

The constrained-efficient enclosure policy regions follow from analyzing the properties of $z_0(t_e) - c \cdot t_e$. As before, the optimal solution depends on the shape of this objective function. Proposition \ref{sbestoe} characterizes the key results.

\begin{proposition}[Second-best Enclosure]\label{sbestoe} A constrained planner who chooses the optimal enclosure rate subject to the labor reaction rule $l_e^o(t_e)$ from (\ref{optle0}) faces a convex objective function in a low-TFP gain economy and a concave one in a high-TFP economy. In the low-TFP economy, the constrained planner chooses:
\begin{itemize}
    \item  \textbf{No enclosure:} $t^c_e=0$, when $z_0(1)-c\leq z_0(0)$ 
    \item  \textbf{Full enclosure:} $t^c_e=1$, when $z_0(1) - c \geq z_0(0)$
\end{itemize}
In the high-TFP economy, the constrained planner chooses:
\begin{itemize}
    \item \textbf{No enclosure:} $t^c_e=0$ when $z_0'(0)\leq c$
    \item \textbf{Full enclosure:} $t^c_e=1$ when $z_0'(1) \geq c$
    \item \textbf{Partial enclosure:} $t^c_e \in (0,1)$ when $z_0'(0) > c$ and $z_0'(1)<c$
\end{itemize} 
\end{proposition} 
\begin{proof}[Proof:]
The shape of the planner's objective function $z_{0}(t_e)-ct_e$ in (\ref{zte0}) depends on $\Lambda_0$. From equation (\ref{condmaxxer}), when $\Lambda_0<1$ (low-TFP), the objective is a quotient of a linear function and a concave decreasing function, making it convex. When $\Lambda_0>1$ (high-TFP), it is a quotient of a linear function and a concave increasing function, making it concave. The optimal policy follows directly from these properties: in the convex case, we compare the function values at endpoints; in the concave case, we can find interior solutions by examining derivatives.
\end{proof}

\begin{figure}[htb!] 
\centering
\includegraphics[width=.9\textwidth]{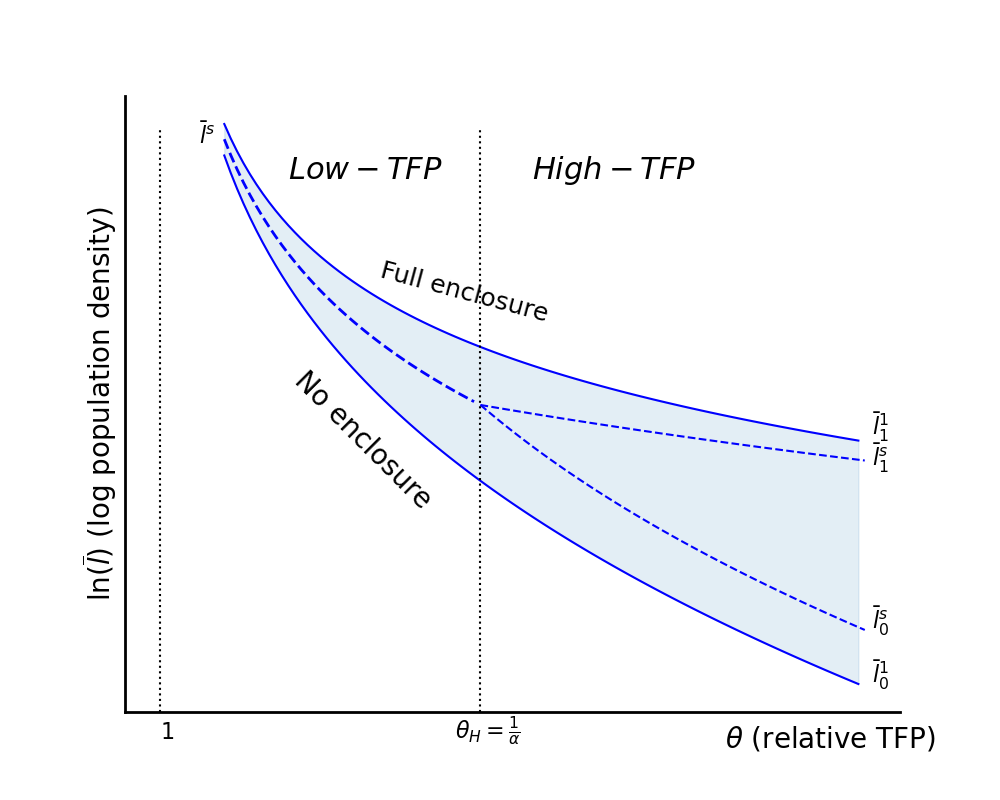}
\caption{First-best and second-best enclosure as functions of population density and expected TFP gain $\theta$.}
\label{fig-social-cond}
\end{figure}  

From Proposition \ref{sbestoe}, in the low-TFP economy, full enclosure yields higher net output when $z_0(1)-c \ge z_0(0)$. This occurs when:

\begin{equation} \label{cond_psconvex}
z_0(1)-c \ge z_0(0)\quad \Leftrightarrow \quad  \bar   l \geq \left[\frac{c}{A}\frac{1}{\theta-1}\right]^\frac{1}{\alpha}=\bar l^s
\end{equation}
where the superscript $s$ denotes second-best decision thresholds. 

In the high-TFP economy, the constrained planner's concave objective yields different thresholds. Enclosure becomes optimal when $z_0'(0)>c$, requiring:

\begin{equation} \label{ps0_locus}
z_0'(0)>c \quad \Leftrightarrow \quad  \bar  l \geq \left[\frac{c}{A}\frac{\alpha }{(\Lambda_0(1+\alpha)-\alpha)} \cdot \frac{1}{(1-\alpha)}\right]^\frac{1}{\alpha}=\bar l_0^s
\end{equation}
Full enclosure becomes optimal when $z_0'(1)>c$:

\begin{equation} \label{l1_star}
z_0'(1)>c \quad \Leftrightarrow \quad  \bar  l \geq \left[\frac{c}{\theta A} \cdot \frac{1}{(1-\alpha)}\right]^\frac{1}{\alpha}=\bar l_1^s
\end{equation}

Between these thresholds, partial enclosure is second-best optimal. Notably, the full enclosure threshold $\bar l_1^s$ in (\ref{l1_star}) coincides with the decentralized threshold $l_1^d$ in (\ref{fullenc_dec_rf}) for high-TFP economies (Proposition \ref{hightfpenc}).

Figure \ref{fig-social-cond} shows the boundaries of second-best enclosure regions (dashed blue lines $\bar l^s$, $\bar l_0^s$, and $\bar l_1^s$) superimposed on the first-best boundaries from Figure \ref{fig-social}. The constrained planner differs from the unconstrained planner in two ways. First, for economies above $\bar l_0^o$ but below $\bar l^s$ and $\bar l_0^s$, the constrained planner avoids initiating enclosures because the resulting labor misallocation would exceed productivity gains. Second, for economies above $\bar l^s$ and $\bar l_1^s$ but below $\bar l_1^o$, the constrained planner chooses full enclosure to eliminate costly labor misallocation, while the first-best planner maintains some unenclosed land.

These second-best benchmarks help identify where decentralized processes lead to inefficient outcomes that policy could improve. Figure \ref{figure6} overlays the constrained planner's decision regions on the decentralized economy's enclosure regions from Figure \ref{figure_private}, with first-best regions shown in lighter shading. Two types of inefficiency emerge. 

First, excessive private enclosure occurs in the red-hatched region (above $\bar l_{gg}^d$ and below $\bar l^s$), where decentralized processes lead to full enclosure despite higher net output under no enclosure.

Second, insufficient enclosure appears in the blue-hatched region (right of $\theta_H$, below $\bar l_1^s$, above $\bar l_0^s$). Within this region, the more finely hatched area between $\bar l_0^s$ and $\bar l_0^d$ represents a coordination failure—no enclosure occurs despite the potential for higher net output. These areas may benefit from policies that subsidize or stimulate enclosure.

\begin{figure}[htb!] 
\centering
\includegraphics[width=.9\textwidth]{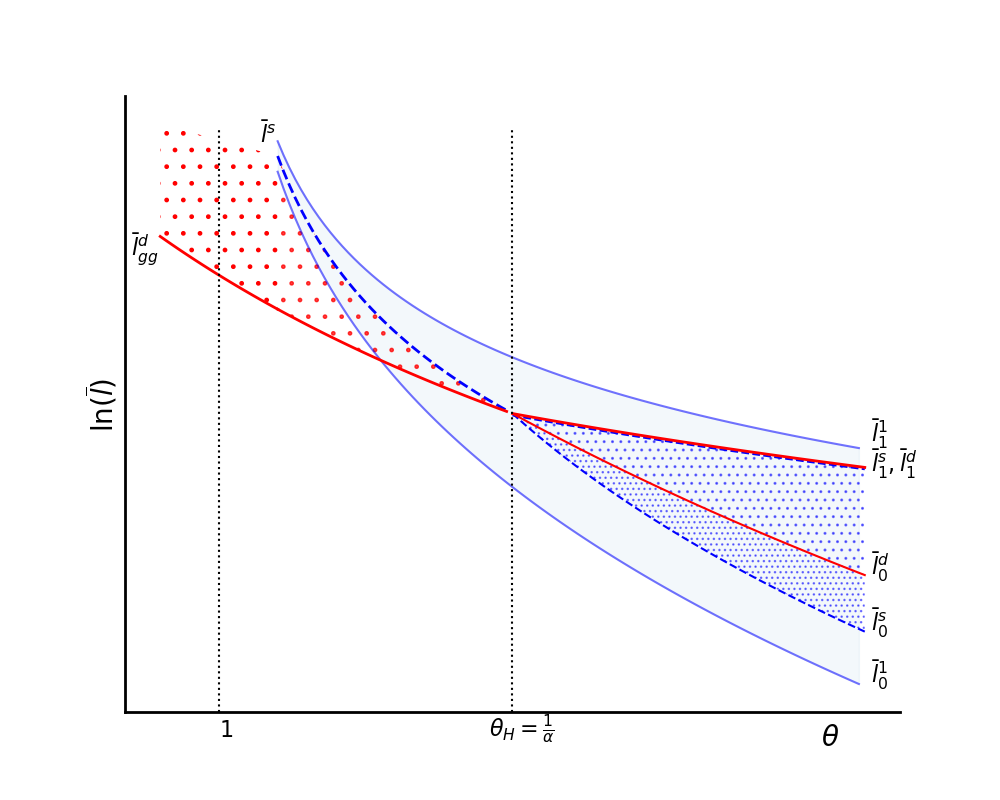}
\caption{Comparing decentralized and second-best outcomes. Red-hatched regions, where $\theta<\theta_H$, show excessive enclosure; blue-hatched regions, where $\theta>\theta_H$, show insufficient enclosure. }
\label{figure6}
\end{figure}  

This analysis reveals two distinct patterns of inefficiency. First, in low-TFP gain economies ($\theta < \theta_H$, below locus $\bar l_0^s$), small parameter changes can trigger sudden inefficient enclosure cascades. An economy just below the $\bar l_{gg}^d$ locus may be efficiently unenclosed, but a slight increase in population density $\bar l$, potential TFP gain $\theta$, or other rent-enhancing factors can precipitate a race to enclose—even when no enclosure remains first-best and second-best efficient. These inefficient property races are particularly stark when $\theta<1$, where enclosure actually reduces productivity. These equilibria represent a form of the `tragedy of the \textit{anticommons}' \citep{buchanan2000}, where excessive enforcement expenditures and private rent capture impede efficient resource use.  This generalizes earlier work on inefficient property races. \citet{anderson1990} identified the possibility of inefficient races to land, while \citet[][p.567]{demeza1992} described the potential for such `cusp catastrophes.' Our global games approach precisely identifies  parametric boundaries where such transitions occur. Similar inefficient races have been documented in patent races \citep{dasgupta1980}.

A qualitatively different inefficiency emerges in high-TFP economies ($\theta > \theta_H$). In the hatched region above $\bar l_0^s$ but below $\bar l_1^d$, coordination failures lead to suboptimal enclosure. Most strikingly, between $\bar l_0^s$ and $\bar l_0^d$, the decentralized economy fails to initiate any enclosure despite clear potential gains in net output.

\subsection{Sources of inefficiency} \label{prelim}

To understand the sources of inefficiency, consider the net social benefits of enclosure from a constrained planner's perspective. The derivative of the planner's second-best objective function, $z_0'(t_e) - c$, reveals two external effects that private enclosers fail to internalize:\footnote{Recall from (\ref{zte0}) that $z_0(t_e)=  [\theta F(t_e, l_e^0(t_e)) + F(1-t_e, 1-l_e^0(t_e))]\cdot A F(\bar T, \bar L)/\bar T$} 

\begin{align}\label{zprimeparts}
     \overbrace{\theta F_T^e A \bar f - c}^\text{net private benefit}
     \underbrace{-F_T^c A \bar f}_\text{external cost} 
     + \overbrace{(\theta F_L^e - F_L^c) A \bar f \cdot \frac{dl_e^0}{dt_e}}^\text{external benefit} 
\end{align}
Where $A \bar f = A F(\bar T, \bar L) / \bar T$ in (\ref{zprimeparts}) stands for potential output per unit of land with no land enclosure. 

Private enclosers consider only the first term—whether market returns exceed enclosure costs ($r-c$). They ignore two external effects. First, the encloser captures rents but without any obligation to pay compensation; enclosure imposes costs on displaced users who lose access to land rents $F_T^c A \bar f$ they previously captured—the `rent-shifting' effect. Second, when $\theta>1$, enclosure creates productivity benefits through reduced labor misallocation, captured by the term $(\theta F_L^e - F_L^c) A \bar f \cdot \frac{dl_e^0}{dt_e}$. This productivity gain occurs as higher-productivity enclosed land replaces lower-productivity open access use.\footnote{Though some labor will be displaced whenever $\theta<\theta_H$, for small marginal adjustments, this displaced labor will have the same marginal product $F_L^c A \bar f$ in the unenclosed sector while remaining workers on newly enclosed plots achieve a higher marginal product $F_L^e A \bar f$. Thus, enclosure increases overall labor productivity when $\theta>1$ but reduces it when $\theta<1$.}

The balance of these externalities determines equilibrium efficiency. When the net of these two external effects is negative, decentralized enclosure processes can lead to inefficiently high levels of enclosure. This occurs in economies in red-thatched regions in Figure \ref{figure6} below the high-TFP gain threshold. On the other hand, if the net of the two external effects is positive, the decentralized economy produces inefficiently low levels of land enclosure (in the blue thatched regions to the right of the high-TFP threshold).

\section{The extended model}\label{extended_model}

The benchmark model follows the literature in making two strong simplifying assumptions: (1) that unenclosed areas are characterized by unregulated open-access and (2) that when enclosures happen, they occur without compensation to existing or potential customary users. We now introduce two parameters, $\mu$ and $\tau$, that allow us to capture more realistic institutional arrangements from the diversity observed historically and across cultures \citep{berry1993, blaufarb2016}. Parameter $\mu \in [0,1]$ captures the degree to which users can regulate access to the commons, while $\tau \in [0,1]$ captures the extent to which stakeholders must be compensated for enclosure. We discuss the interpretation of each of these parameters in the next two subsections and then integrate them into the model to deliver new general equilibrium predictions.

\subsection{The regulated commons}\label{policy_regulate}

Real-world customary tenure regimes often differ markedly from pure open access. Communities may find ways to regulate access to resources through membership restrictions, land use rules, or access fees.\footnote{Access fees may be explicit charges or implicit costs, as when newcomers must prove themselves worthy of permanent rights.} Community members may also develop individualized plot rights that at least partly overcome the risks of use-it-or-lose-it possession, allowing them to maintain rights while pursuing outside opportunities. Local markets for exchange of land and labor among insiders may develop, and even conditioned transfers to outsiders\citep{shipton1992, boone2014, onoma2009}. 

At an equilibrium we can think of possession rights $\frac{T_c}{L_c}$ becoming relatively secure and stable. Given the linear homogeneous production technology, while individual households could be self-sufficient, we could also describe the situation as one where local factor markets emerge with wage $w_c=MP_L^c$ and land rental $r_c=MP_T^c$.\footnote{This provides an idealized benchmark. Local markets may face distortions from information asymmetries and property disputes \citep{chen2023, goldstein2008}.}

To capture some of this institutional complexity, we modify the equilibrium labor allocation condition (\ref{labmkteq2}) that governs entry and exit from the commons. With local factor markets we can substitute $w_c=MP_L^c$ and $r_c=MP_T^c$ and write:
 
\begin{equation*} 
w_e + \mu r_c \frac{T_c}{L_c} = w_c + r_c \frac{T_c}{L_c}
\end{equation*}
which can be rewritten as:
\begin{equation} \label{mu_move2}
w_e - w_c = (1-\mu) \cdot r_c \frac{T_c}{L_c}
\end{equation}
Parameter $\mu$ can be interpreted as an access fee for entering the commons or, equivalently, as a measure of customary rights' security and transferability. New entrants must pay for access, retaining only fraction $(1-\mu)$ of possession rents. For established rights-holders who leave for formal sector opportunities, the interpretation is symmetric: in pure open access ($\mu=0$), abandoning possession means losing all rents, while with $\mu>0$, leavers can retain (or cash out) fraction $\mu$ of their established rights.\footnote{For example, leavers might rent their customary land to insiders but face the risk that subtenants squat and refuse payment with probability $(1-\mu)$.}

Under either interpretation, expression (\ref{mu_move2}) shows that labor moves between sectors until the net capturable rents equal foregone opportunities. When $0\leq\mu<1$, imperfect regulation allows some labor crowding to persist. At $\mu=1$, access fees fully internalize entry externalities, equalizing marginal products. However, this perfect regulation differs from enclosure when $\theta\neq1$, as it neither transforms technology nor protects against displacement by potential enclosers.

\subsection{Power and Compensation}\label{power}

The benchmark model assumes enclosers are outsiders who can capture new land rents without compensating displaced users. While this might describe cases where settlers or elites dispossess locals through state violence, it excludes many important scenarios. Enclosers might be insiders with existing claims, outsiders required by state policy to compensate locals, or parties in contested property situations—as when landlords seek to convert customary tenant to commercial leases.

We introduce parameter $\tau \in [0,1]$ to measure required compensation to displaced stakeholders. For each enclosed unit of land, enclosers must pay $\tau r_c$ to existing users, where $r_c=MP_T^c=F_T^cA\bar f$ represents the land rent these users captured in the unenclosed sector (see equation \ref{zprimeparts}). The level of $\tau$ reflects groups' relative bargaining power and ability to resist displacement—what some might call the `balance of class forces'.

With compensation, enclosure is privately profitable when:

\begin{equation} \label{tau_enclose}
 (\theta F_T^e A \bar f - c) - \tau F_T^c A \bar f \ge 0 
\end{equation}

When $\tau>0$, enclosers are forced to internalize part or all of the external costs they impose on others as earlier identified in (\ref{zprimeparts}). This reduces the rent-seeking motive that led to excess enclosure. When $\tau=1$ (full compensation) and $\mu=1$ (full regulated access to the commons that eliminates labor misallocation), private incentives become aligned with social optimality.

The role of compensation connects to broader debates in economic history. Political Marxists like \citet{brenner1976, wood2002} emphasize how sudden shifts in class power—a sharp reduction in $\tau$ from a regime that protected peasant customary rights to one that allowed eviction and rack rents—drove England's agricultural transformation.  More generally, $\tau$ allows us to analyze various historical cases—from revolutionary France to modern land formalization programs—where the degree of required compensation shaped property regime transitions \citep{blaufarb2016}.

\subsection{Policy and the second-best} \label{policy_both}

With this extended version of the model, we can now analyze how varying degrees of commons regulation ($\mu$) and stakeholder protection ($\tau$) affect equilibrium outcomes. For labor allocation with regulated access ($\mu \in [0, 1]$), our parametric model yields:

\begin{equation} \label{mulabor}
l_e^\mu(t_e) = \frac{\Lambda_\mu t_e }{(1+(\Lambda_\mu-1)t_e)}   \quad \textrm{where}\quad\Lambda_{\mu}=\left(\frac{\alpha\theta}{1-\mu\cdot (1-\alpha)}\right)^\frac{1}{1-\alpha}
\end{equation}
The high-TFP threshold from Definition 1 is now modified to become:

\begin{equation}\label{mutheta}
   \theta_H^\mu=\frac{1}{\alpha} -\mu \cdot \frac{1-\alpha}{\alpha}
\end{equation}

As commons access becomes more regulated, labor misallocation decreases. As $\mu$ increases from 0 to 1, $\Lambda_\mu$ moves from $\Lambda_0$ in equation (\ref{optle0}) to the optimal $\Lambda_1$ from equation (\ref{mplfoc}), while the high-TFP threshold falls toward 1. For any given $\theta$, this threshold decreases monotonically with $\alpha$.

To analyze how incentives change with regulation ($\mu$) and compensation ($\tau$), we substitute $l_e^\mu(t_e)$ from (\ref{mulabor}) into (\ref{eqr_rf}):

\begin{equation} \label{reopt}
r^e_\mu(t_e)=\theta(1-\alpha) A\bar l^\alpha \cdot \left(\frac{\Lambda_{\mu}}{1+(\Lambda_{\mu}-1)t_e}\right)^\alpha 
\end{equation}
We can derive an analogous expression for $r_{\mu}^c(t_e)$ which determines the implied rental rate of land within the unenclosed sector. The encloser now encloses a unit of land only when

\begin{equation} \label{reoptminuss}
r^e_\mu(t_e)-\tau \cdot r_\mu^c(t_e)-c=(1-\alpha) A \bar l^\alpha \cdot \frac{\theta\Lambda^\alpha_\mu-\tau}{(1+(\Lambda_\mu-1)t_e)^\alpha } -c \ge 0
\end{equation}
Note that when an existing customary user considers enclosing their own plot, they effectively face $\tau=1$ since they must weigh the new enclosed rental value against the customary rents they currently capture through possession.

With these modifications, we can adjust Propositions 3 and 4 with appropriately modified boundary loci. Figure \ref{figure4x4} illustrates how these modifications affect equilibrium outcomes. The red-shaded regions show decentralized enclosure decisions under different combinations of $\mu$ and $\tau$, while blue-shaded regions show unchanged socially optimal decisions. Panel (a) shows our benchmark case: unregulated commons ($\mu=0$) with no compensation ($\tau=0$).

\begin{figure}[htb!] 
\centering
\includegraphics[width=\textwidth]{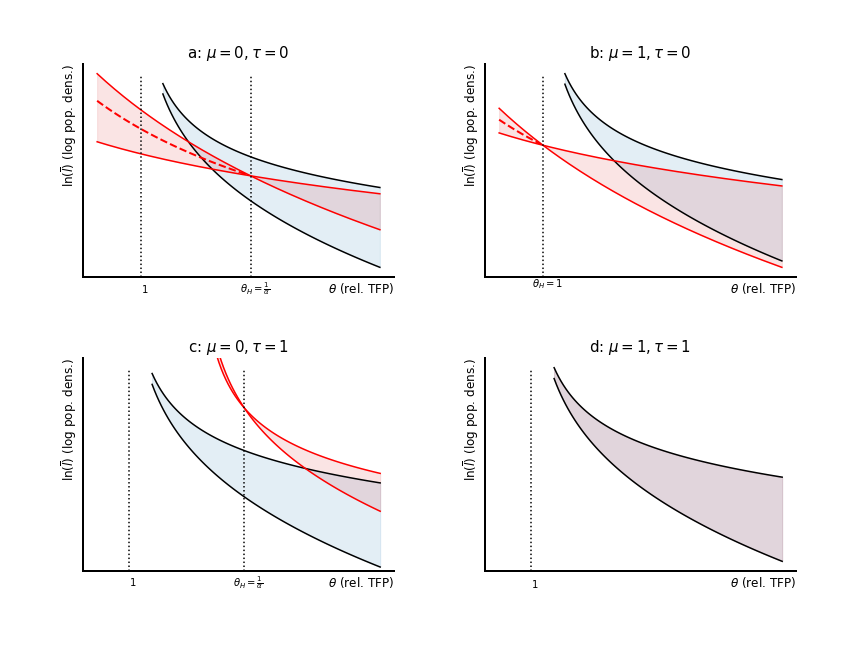}
\caption{Comparing $\mu$ and $\tau$ policy regimes. (a) Benchmark case with an unregulated commons and zero compensation to those displaced by enclosures; (b) A fully regulated commons but zero compensation;  (c) Unregulated commons but full compensation; and (d) Regulated commons and full compensation where private enclosure decisions now coincide with socially optimal ones.}
\label{figure4x4}
\end{figure}

From equation (\ref{reoptminuss}), enclosure returns increase with $\mu$ through higher $\Lambda_\mu$. Better commons regulation increases labor supply to the enclosed sector, lowering wages and raising enclosure returns since operators no longer need to pay high premiums to compensate labor for abandoning possession. Panel (b) shows this case of well-regulated commons ($\mu=1$) without compensation ($\tau=0$): while inefficient under-enclosure disappears, the range of inefficient over-enclosure expands. However, higher $\mu$ lowers the threshold $\theta_H^\mu$ below which enclosure decisions exhibit strategic complementarity, making catastrophic property cascades less likely in high-TFP economies (when $\mu=1$, $\theta_H^{\mu=1}=1$).

Panel (c) shows the opposite case: unregulated commons ($\mu=0$) but full compensation ($\tau=1$). The compensation requirement prevents rent-seeking inefficient enclosure but expands the range of inefficient under-enclosure. This creates an interesting tension—weaker property rights ($\tau<1$) may generate higher output through increased technological transformation, albeit at the cost of greater inequality. We see this in high-TFP economies that would efficiently fully enclose under $\tau=0$ (panel a) but fail to enclose under $\tau=1$ (panel c).

Finally, panel (d) shows the first-best case where both regulation and compensation are complete ($\mu=\tau=1$). Here, decentralized decisions align perfectly with social optimality. The contrast between panels (b) and (c) illustrates the theory of second-best: in a general equilibrium setting, reducing misallocation in one market may worsen outcomes if other markets remain distorted.

\section{Extensions and applications}\label{sec-extensions}

We now show how the extended model can be used to analyze four key political economy ande development topics: power shifts in contested claims (Section \ref{who_encloses}), labor release and wages (Section \ref{labor_release}), frontier land policy (Section \ref{sec-encompass}), and structural transformation (Section \ref{manuf_sec}).

\subsection{Power shifts and contested claims}\label{who_encloses}

In many historical contexts, landlords and tenants often contested land rents and eviction protections under customary tenure regimes. Such contestation emerged in many systems that evolved from manorial or colonial estates \citep{binswanger1995}. As serfdom declined in England after the fourteenth century, peasants acquired valuable customary rights including inheritance, protection from eviction, and limits on rent increases. While landlords sought to challenge these arrangements, peasants resisted, making property rights dependent on both local political conditions and economic factors like commercialization and population pressure \citep{hatcher2001,allen1992}.

Landlords and tenants contested the terms of customary rents and eviction protections. To model the effects of a shift in bargaining power, assume  that tenants initially pay only fraction $\mu$ of customary land's marginal return ($F_T^c A \bar f$). Through enclosure, landlords could evict tenants and charge the full market rent ($\theta F_T^e A \bar f$) to new or existing tenants. However, landlords must weigh these potential gains against three costs: direct enclosure costs ($c$), foregone customary rents ($\mu F_T^c A \bar f$), and compensation to displaced tenants ($\tau F_T^c A \bar f$, where $\tau \ge \mu$). Strong tenant protections are reflected in a high value of $\tau$.

A shift in political power favoring landlords—perhaps through capture of state institutions \citep{marx1992, brenner1976, wood2002}—might reduce required compensation $\tau$. This could make previously unprofitable enclosures viable, with consequences that depend on the economy's position relative to threshold $\theta_H^{\mu}$. In low TFP gain environments, reduced protection in the form of reduced compensation could trigger destructive enclosure races; in high TFP gain environments, the shift helps overcome coordination failures that had been blocking or limiting the extent of productive enclosure. Figure \ref{figure4x4} illustrates these transitions. When customary arrangements are weakly regulated ($\mu=0$), reducing compensation (from $\tau=1$ to $\tau=0$) moves the economy from panel c to panel a. However, when customary arrangements are well regulated ($\mu=1$), the same reduction in customary user protections (panel d to panel b) increases the range of environments where excessive and inefficient enclosure takes place.

This framework helps interpret competing historical accounts of English enclosures. \citet[][p. 937]{allen1982} has argued that while the consequences of enclosures remain debated, the ``the proximate cause the English enclosure has always been clear...[they were] invariably initiated by landowners because they expected their tenant farmers would pay higher rents after the Parish was enclosed.'' In his interpretation landlords initiated enclosures to capture higher rents after a century of observing productivity improvements in the open-field system. In our model this would be equivalent to increased $A$, which makes enclosure more profitable through reduced $c/A$. This interpretation aligns also with \citet[p. 77]{clark1998}, who suggests commons persisted because they were ``not very inefficient.'' These are however by no means agreed interpretations.  Recent empirical work by \citet{heldring2024} finds evidence for a causal relationship between Parliamentary enclosures and agricultural productivity and structural change, although like \citet{allen1992} they find evidence of labor displacement. \citet{lazuka2023a} point to similar causal evidence for Sweden.  

The U.S. frontier provides a contrasting example of how policy shaped enclosure parameters. The Preemption Acts of the 1830s and 1840s effectively increased compensation requirements ($\tau$) by recognizing squatters' rights and improvements, while claim clubs raised commons regulation ($\mu$) through local organization \citep{murtazashvili2013}. The U.S. government directly set enclosure costs ($c$) through residency requirements and land purchase prices. Through these and later Homestead Acts, policy created an ordered progression from squatting to formal ownership, though with starkly different compensation requirements—high $\tau$ for early settlers but effectively zero for displaced Native Americans \citep{carlos2022a}.

The frontier experience in the United States illustrates both successful and problematic property transitions. \citet{anderson2004, libecap1993, desoto2000} document how local communities developed efficient property arrangements adapted to their environments, often leading to orderly formalization. Technological changes, like the development of low-cost barbed wire (reducing $c$), led to demands to convert from common to private property (which had often been the preferred property regime by cattlemen's associations). However, top-down policies like the 1862 Homestead Act sometimes triggered inefficient outcomes—premature settlement, wasteful property races, and displacement of existing rights holders particularly at the cost of uncompensated Native American displacement \citep{carlos2022a}.

\subsection{Labor release and wages}\label{labor_release}

Earlier literature, notably \citet[p. 225]{weitzman1974}, argued that labor ``will always be better off with (inefficient) free access rights than under (efficient) private ownership,'' and highlight the labor displacement potential of enclosures.  Marxian interpretations of the English enclosure movements often stress the labor displacement \citep{tawney1912,humphries1990,cohen1975,perelman2000}. But other scholars question these interpretations \citep{whittle2013,shaw-taylor2001}. Our framework cannot settle what is ultimately an empirical debate, but it helps frame the theoretical debate by nesting differing interpretations and revealing the conditions under which this mechanism might be more likely to be at work.  Total labor earnings can be written as:

\begin{equation} \label{labinc}
Y_L = \alpha \cdot \theta \cdot A F(\bar T, \bar L)  \cdot 
\left ( \frac{1+(\Lambda_{\mu}-1)t_e}{\Lambda_{\mu}}  \right )^{1-\alpha}
\end{equation}
The distributional effects of enclosure depend critically on the economy's characteristics. In an unenclosed economy ($t_e=0$), labor captures all output ($Y_L = AF(\bar T, \bar L)$) through wages and possession rents. Under full enclosure ($t_e=1$), labor earns only the neoclassical share ($Y_L = \alpha \cdot \theta A F(\bar T, \bar L)$). While total output rises whenever $\theta>1$, labor only gains when $\theta>\theta_H^{\mu}$—when productivity gains sufficiently boost labor demand and wages to offset lost possession rents.

Other parameters shape these outcomes. Higher compensation requirements ($\tau$) reduce inefficient enclosures that would lower wages (Figure \ref{figure4x4}, panel c). Better commons regulation (higher $\mu$) lowers the threshold $\theta_H^{\mu}$ above which labor benefits, though it may increase enclosure pressure.

In a very different context, Mexico's 1990s ejido reforms illustrate these mechanisms. When ejido households could obtain ownership certificates, breaking the link between land use and rights, significant labor release followed. \citet{dejanvry2015a} found certified households became 28 percent more likely to have migrant members, with larger effects where property rights were initially less secure—consistent with our model's predictions for cases where $\theta<\theta_H$.

\subsection{Encompassing interests and frontier colonization}\label{sec-encompass}

Enclosure processes often involve actors with more unified and encompassing interests than atomistic decision-makers. Colonial governments and their agents, in particular, frequently controlled enclosure policy to advance specific objectives or group interests. In the early and mid 19th century Edward G. Wakefield's `systematic colonization' theory provided a framework for approaching frontier land policy that also proved influential in shaping classical political economy.

Wakefield developed an influential theory of colonial settlement that highlighted a fundamental tension in land-abundant frontiers: when laborers could easily claim independent holdings, wage-labor operations became unprofitable due to high wages and what he termed `inconstancy of labour' \citep{wakefield1849}. His proposed solution, which shaped British colonial policy, called for authorities to establish strong preemptive rights over `wastelands' and sell land at a `sufficient price' to slow laborers' transition to ownership, thereby ensuring concentrated settlements and a reliable, low-wage workforce. Karl Marx later seized on Wakefield's colonial analysis, arguing in the final chapter of Capital that it had inadvertently revealed a broader truth about capitalist production: that capital accumulation required deliberately limiting workers' access to independent livelihoods to create a dependent workforce \citep[][chapter 33]{marx1992}.

These policies had varying success. In settler colonies like South Australia and New Zealand, and the United States and Canada,  while displacing native peoples, organized squatter pressure prevented implementation of Wakefield's high `sufficient price.' However, more extreme versions of severely restricting smallholder access to create wage labor pools were implemented in parts of Sub-Saharan Africa and Latin America \citep{binswanger1993, solberg1969, legrand1984}.

To formalize these ideas about encompassing interests, consider a syndicate—like Wakefield's proposed colonization companies—that receives rights over unenclosed land. The syndicate acts as a monopolistic encloser, aiming to maximize returns from land enclosure through sales or rentals to competitive farm operators.\footnote{We abstract from potential labor monopsony effects that might arise if the encloser also became a major employer. See \citet{conning2007} for analysis of land monopoly-labor monopsony interactions and related literature.} For simplicity, we assume unenclosed areas remain under free access ($\mu=0$) and enclosers pay no compensation ($\tau=0$).

The syndicate chooses an enclosure rate $t_e$ to maximize land rents net of enclosure costs:
\begin{equation}\label{landrets}
    \pi(t_e)=r(t_e) \cdot t_e-c \cdot t_e 
\end{equation}
where $r(t_e)$ represents the market rental rate in the enclosed sector from (\ref{eqr_rf}). In low-TFP economies, $\pi(t_e)$ is convex, leading the monopolist to either fully enclose or not enclose at all. Full enclosure occurs when $\pi(1)>0$, or when:
\begin{equation} \label{pm0}
\pi(1)>0\quad\Leftrightarrow\quad\bar  l \geq \left[\frac{c}{\theta} \cdot \frac{1}{(1-\alpha)}\right]^\frac{1}{\alpha}
=\bar l^m
\end{equation}
The condition (\ref{pm0}) exactly matches the boundary condition $\bar l_1^d$ we saw earlier in the decentralized case in (\ref{fullenc_dec_rf}), as characterized in Proposition \ref{lowtfpenc}. However, while in the atomistic case, economies above this $l_1^d$ boundary but below the $l_0^d$ boundary of (\ref{noenc_dec_rf}) faced a situation where $r(1)>c>r(0)$, implying multiple equilibria at either $t_e=0$ or $t_e=1$. 
 The monopolist, however, internalizes these effects and will jump directly to the more profitable full enclosure allocation $t_e=1$ in a range of environments where atomistic enclosure processes would not, even though it is socially inefficient.  This aligns with Wakefield's vision: to curtail labor's access to unenclosed land to create an increase in labor supply that will reduce wages and raise land rents in the enclosed sector.

\begin{figure}[htb!] 
\centering
\includegraphics[width=.9\textwidth]{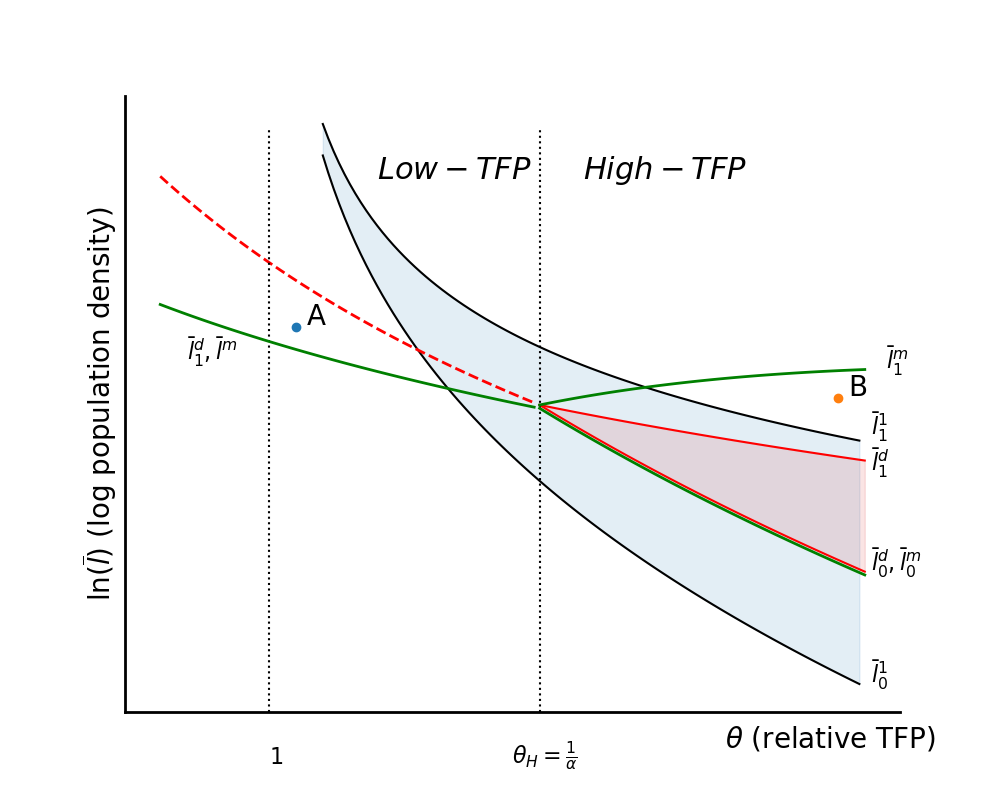}
\caption{Monopolistic versus competitive and optimal enclosure. The monopolist encloses fully above the thick green line in low-TFP regions, but chooses partial enclosure between $\bar l_0^m$ and $\bar l_1^m$ in high-TFP regions.}
\label{figure3}
\end{figure}
The monopolist's strategy differs markedly in high-TFP environments ($\theta \geq \theta_H$). Here, the monopoly encloser changes strategy and now holds land off the market in circumstances where the decentralized market and the social planner would enclose. The monopolist's choice depends on two conditions. First, enclosure begins when $\pi'(0)>0$, where:
\begin{equation} \label{landrets_d}
    \pi'(t_e) = r(t_e) +r'(t_e) t_e -c 
\end{equation}
At $t_e=0$, this yields the same threshold for starting enclosures as the decentralized economy ($\bar l_0^d$ from (\ref{noenc_dec_rf})) developed pursuant to Proposition \ref{hightfpenc}, which we reproduce and relabel here:
\begin{equation} \label{pm0_locus}
\pi'(0)'\geq 0\quad\Leftrightarrow\quad\overline{l} \geq \frac{1}{\Lambda_0}
\left[\frac{c}{\theta} \cdot
\frac{1}{(1-\alpha)}\right]^\frac{1}{\alpha}=\bar l_0^m
\end{equation}
in this region the monopolist begins to enclose at the same time as the decentralized economies do. However, as population density increases, the incentives of the monopolist diverge from those of decentralized enclosers because the monopolist internalizes the effect of increasing the enclosure rate $t_e$ on the rental rate $r(t_e)$. Full enclosure occurs only when $\pi'(1)\geq 0$:

\begin{equation} \label{pm1_locus}
\pi'(1)\geq 0\quad\Leftrightarrow\quad   \bar l\ \geq \left[\frac{\alpha c}{(1-\alpha)}\frac{\Lambda_0}{\Lambda_0 (1-\alpha)+\alpha}\right]^\frac{1}{\alpha}=\bar l_1^m
\end{equation}
Figure \ref{figure3} illustrates how monopolistic enclosure patterns differ from decentralized and socially optimal outcomes. Two distinct patterns emerge:

In low-TFP regions, a monopolist does not face strategic uncertainty about whether other parcels will also be enclosed, so it chooses to move to full enclosure in every situation where $r(0)<c$ yet $r(1)>c$. At a point such as $A$, the monopolist pushes to inefficiently enclose all land where the decentralized economy would not. This pattern aligns with Wakefield's vision of using land policy to force up labor supply.

In high-TFP regions ($\theta \geq\theta_H$), however, the monopolist behaves quite differently. At point $B$, it chooses only partial enclosure even when both decentralized processes and social planners would fully enclose. The monopolist restricts enclosure because further institutional and technological transformation in the unenclosed sector would raise wages, reducing land rents.  This fits Alain de Janvry's \citeyearpar{dejanvry1981} `functional dualism' argument for parts of Mexico and Latin America, where commercial farmers benefited from maintaining a low-productivity peasant sector with insecure property rights because it ensured a continued cheap labor supply.

\subsection{Structural transformation and manufacturing}\label{manuf_sec}

Historical interpretations of English enclosures' role in structural transformation remain contested. \citet{allen1992} characterizes both Marxist \citep{brenner1976} and `Tory triumphalist' historians \citep{chambers1953} as `agrarian fundamentalists' who saw customary institutions as inherently backward, requiring removal for agricultural modernization and labor release which they believe helped jump start the industrial revolution. Allen's \citeyearpar{allen1992} evidence from the South Midlands suggests instead that yeoman farmers achieved significant efficiency and productivity gains under the open field system, implying Parliamentary enclosures were primarily redistributive—helping landlords capture rising rents rather than driving efficiency improvements. While finding evidence of labor displacement, Allen argues this initially created rural underemployment rather than industrial labor supply. Other interpretations, including recent work by \citet{heldring2024}, however, suggests enclosures did contribute to both agricultural productivity and industrial growth.

We cannot settle these empirical debates, but our theoretical framework helps clarify how endogenous changes in property regimes shape structural transformation. Earlier models like \citet{cohen1975} and \citet{crafts2004}, and recent macro-misallocation papers \citep{chen2017a, chen2023, gottlieb2019}, study these issues but treat property rights transformation as exogenously driven. By modeling institutional change as endogenous, we can analyze both the conditions that trigger property regime transitions and the general equilibrium feedback effects that shape their timing and extent that these other approaches miss.

While a complete analysis of structural transformation lies beyond our scope, we can extend the model to understand how property regime transitions interact with its development. Adding a manufacturing sector that competes for labor, the labor market balance condition becomes:

\begin{equation}
l_e + l_c = 1 - l_m
\end{equation}
where $l_m$ is the share of total labor $\bar L$ employed in manufacturing.

Manufacturing employs constant returns technology $G(K,L)=A_m \cdot K^{1-\beta} L^\beta$, with productivity $A_m$ and sector-specific capital $\bar K$. With $L_m=l_m\cdot \bar L$, manufacturing output is:

\begin{equation*}
 G(\bar K, L_m)  =   l_m ^{\beta} \cdot G(\bar K, \bar L)
\end{equation*}

In this small open economy, the relative price of manufactures $p$ is set by world markets. Labor mobility equalizes returns across the three sectors:

\begin{equation*}
MP_L^m = 
p \cdot\beta \cdot \left ( \frac{1}{l_m} \right ) ^{1-\beta} \cdot  \bar k^{1-\beta}
\end{equation*}
The value average product of labor in the customary sector becomes:
\begin{equation*}
\begin{aligned}
AP_L^c  &=  A \left ( \frac{1-t_e}{(1-l_m)-l_e} \right ) ^{1-\alpha} \cdot  \bar t^{1-\alpha}
\end{aligned}
\end{equation*}
Labor mobility ensures no incentive to move across sectors: 
\begin{equation*} \label{maneq}
w = p \cdot MP_L^m = MP_L^e = AP_L^c>MP_L^c
\end{equation*}

Solving $MP_L^e = AP_L^c$ for $l_e$ yields a modified expression for enclosed sector labor demand:
\begin{equation}
l_e^0(t_e) = \frac{\Lambda_0 t_e }{(1+(\Lambda_0-1)t_e)} \cdot (1-l_m)  
\end{equation}
which adjusts our earlier result (\ref{optle0}). Other expressions, including the private return to enclosure $r(t_{e})$ in (\ref{eqr_rf}), follow similarly.

Manufacturing opportunities increase labor's opportunity cost to remaining in the customary sector.  We could call this a potential `pull' effect. There are also potential labor `push' effects. The elasticity of labor supply to manufacturing depends on the agricultural production parameters ($A, \alpha, \theta$) and policy environment ($p, \bar l, c, \mu, \tau$) which determine whether land will become enclosed and labor released.  This framework reveals how structural transformation might be constrained by failure to transform property relations or accelerated by inefficient enclosure and urban migration. The framework is thus able to generate, depending on initial conditions, either smooth efficient transformation, to get stuck at insufficient transformation with excess amounts of labor trapped in the customary, as well as potential premature or excess enclosure and structural transformation.

\section{Conclusions} \label{conclusion}

Our analysis provides a unified framework for understanding when and how property regimes transition from open access or customary arrangements to more exclusive forms. By incorporating technological change ($\theta$), enclosure costs ($c$), commons regulation ($\mu$), and compensation policies ($\tau$), we extend earlier theoretical work to reveal how these factors jointly determine equilibrium outcomes with endogenous transformation. The framework illuminates when property regime transitions promote efficiency and when they primarily redistribute rents, and under what conditions decentralized processes might lead to either coordination failures or inefficient property races.

Our framework offers insights for contemporary land formalization efforts. Recent decades have seen widespread attempts to rewrite national land laws and implement property mapping, registration, and titling programs \citep{aldenwily2018}. While driven partly by international organizations promoting market reforms and state capacity building \citep{desoto2000, deininger2003}, and enabled by new mapping technologies, these programs often reflect governments trying to catch up with transformations of customary regimes already underway.

As James C. \citet{scott1998} notes, local property arrangements often evolve independently of state plans, creating complex social forms that can be opaque to outsiders. While this opacity sometimes protects communities from predation, it may also create opportunities for local elites to block beneficial transitions or insert themselves to exploit them \citep{onoma2009}. Our model reveals how formalization efforts might either enhance efficiency or trigger destructive property scrambles, depending on local conditions. In low-density, low-productivity environments, hasty formalization (reducing $c$) might spark inefficient property races that displace existing rights holders. A more prudent approach would first recognize community ownership  and strengthen local governance capacity (align $\mu$ and $\tau$), in ways that may support eventual individual titling. 

National legislation has moved in this direction  by increasingly recognizing rural communities as collective owners and acknowledged customary law \citep{aldenwily2018}. This reflects lessons from past failed top-down formalization attempts and marks a significant shift from past practices where states either claimed trusteeship over land without protecting existing users or, worse, declared it \textit{terra nullius} to facilitate transfers to privileged groups. Current programs take a more balanced approach, working to identify both community and individual claims, resolve disputes and establish village boundaries and formalize their rights through documented titles, often before moving on to individualization efforts \citep{takeuchi2022,deininger2003}.

Our framework reveals critical policy tradeoffs in property regime transitions. In low-density, low-TFP environments, customary regimes may generate minimal inefficiencies relative to formalization costs. Here, well-intentioned efforts to reduce formalization costs ($c$) may be counter productive and risk triggering unecessary and destructive property scrambles. The wiser approach begins with recognizing community ownership and strengthening local governance before considering individual formalization.

Different challenges emerge in high-TFP environments ($\theta>\theta_H$), where coordination failures can block beneficial transitions. Arguably, this represents the reality in many developing countries, where the adoption of new technologies, crops, or larger investment projects might be impeded by the persistence of non-exclusive and non-transferable forms of land property ownership.  While reducing formalization costs might help overcome some of these failures, success still requires strong local governance and compensation mechanisms to protect existing rights holders.

These insights are particularly relevant as population pressure and economic growth strain traditional institutions \citep{holden2013,baland1996}. While many policymakers advocate privatization as a top-down solution our bottom-up equilibrium analysis, focused on decentralized processes and the strategic choices and reactions of both enclosers and those they might displace, reveals more nuanced understanding of different pathways and their efficiency and distributional consequences. In high-TFP economies, enclosure can increase labor intensity on private parcels, reducing pressure on commons. However, in low-TFP settings, each enclosure displaces labor onto remaining commons, potentially accelerating rather than arresting environmental degradation. Seemingly stable environments can suddenly tip into wasteful property races that concentrate ownership while crowding the displaced onto fragile lands.

This complexity reminds us that customary regimes are not inherently inefficient. Property institutions must evolve with their environment, but transitions require careful attention to local conditions and potential distributional consequences.

\newpage
%\bibliographystyle{apa}
%\bibliography{references}
\printbibliography

@book{dasgupta1979,
  title = {Economic Theory and Exhaustible Resources},
  author = {Dasgupta, Partha S. and Heal, Geoffrey M.},
  date = {1979},
  publisher = {Cambridge University Press},
  url = {https://books.google.com/books?hl=en&lr=&id=CaU_FXSzk0AC&oi=fnd&pg=PR7&dq=economic+theory+and+exhaustible&ots=bmKD168k7I&sig=JMNLG_fDbyKXkZWz-st9F6UD_Tg},
  urldate = {2025-01-06}
}

@article{lazuka2023a,
  title = {The {{Causal Effects}} of {{Enclosures}} on {{Production}} and {{Productivity}}},
  author = {Lazuka, Volha and Bengtsson, Tommy and Svensson, Patrick},
  date = {2023},
  journaltitle = {IZA Discussion Paper},
  number ={16394},
  url = {https://papers.ssrn.com/sol3/papers.cfm?abstract_id=4547698},
  urldate = {2025-01-05}
}

@article{lerner1972,
  title = {The {{Economics}} and {{Politics}} of {{Consumer Sovereignty}}},
  author = {Lerner, Abba P.},
  date = {1972},
  journaltitle = {The American Economic Review},
  volume = {62},
  number = {1/2},
  eprint = {1821551},
  eprinttype = {jstor},
  pages = {258--266},
  publisher = {American Economic Association},
  issn = {0002-8282},
  url = {https://www.jstor.org/stable/1821551},
  urldate = {2024-11-28}
}

@article{carlos2022a,
  title = {Indigenous Nations and the Development of the {{US}} Economy: {{Land}}, Resources, and Dispossession},
  shorttitle = {Indigenous Nations and the Development of the {{US}} Economy},
  author = {Carlos, Ann M. and Feir, Donna L. and Redish, Angela},
  date = {2022},
  journaltitle = {The Journal of Economic History},
  volume = {82},
  number = {2},
  pages = {516--555},
  publisher = {Cambridge University Press},
  url = {https://www.cambridge.org/core/journals/journal-of-economic-history/article/indigenous-nations-and-the-development-of-the-us-economy-land-resources-and-dispossession/E2A9A57E9F58D7CBA2445B74E5E4FB8B},
  urldate = {2024-10-24}
}

@book{anderson2004,
  title={The Not So Wild, Wild West: Property Rights on the Frontier},
  author={Anderson. and Hill},
  isbn={9780804748544},
  lccn={2004001038},
  series={Stanford economics and finance},
  url={https://books.google.com/books?id=A7727zJQ51IC},
  year={2004},
  publisher={Stanford Economics and Finance}
}

@article{demsetz1967,
  title = {Towards a theory of property rights},
  author = {Harold Demsetz},
  date = {1967},
  journaltitle = {American Economic Review},
  volume = {57}, 
  number = {2},
  pages = {347--359}
}

@article{shaw-taylor2001,
  title = {Parliamentary Enclosure and the Emergence of an {{English}} Agricultural Proletariat},
  author = {Shaw-Taylor, Leigh},
  date = {2001},
  journaltitle = {The Journal of Economic History},
  volume = {61},
  number = {3},
  pages = {640--662},
  publisher = {Cambridge University Press}
}

@article{hafer2006,
  title = {On the {{Origins}} of {{Property Rights}}: {{Conflict}} and {{Production in}} the {{State}} of {{Nature}}},
  shorttitle = {On the {{Origins}} of {{Property Rights}}},
  author = {Hafer, Catherine},
  date = {2006},
  journaltitle = {The Review of Economic Studies},
  shortjournal = {The Review of Economic Studies},
  volume = {73},
  number = {1},
  pages = {119--143},
  issn = {0034-6527},
  doi = {10.1111/j.1467-937X.2006.00371.x},
  url = {https://doi.org/10.1111/j.1467-937X.2006.00371.x},
  urldate = {2023-11-22}
}

@book{augustinus2003,
  title = {Handbook on Best Practices, Security of Tenure, and Access to Land: Implementation of the {{Habitat Agenda}}},
  shorttitle = {Handbook on Best Practices, Security of Tenure, and Access to Land},
  author = {Augustinus, Clarissa},
  date = {2003},
  publisher = {UN-Habitat},
  url = {https://books.google.com/books?hl=en&lr=&id=21lzhCMan5cC&oi=fnd&pg=PA1&dq=UN+habitat+land+best+practices+Augustinus&ots=cnn20kK_vW&sig=YnjunmD2mpKnbHjFWlZL3bM7SlM},
  urldate = {2024-04-18}
}

@book{barbier2010,
  title = {Scarcity and Frontiers: How Economies Have Developed through Natural Resource Exploitation},
  shorttitle = {Scarcity and Frontiers},
  author = {Barbier, Edward B.},
  date = {2010},
  publisher = {{Cambridge University Press}},
  url = {https://books.google.com/books?hl=en&lr=&id=yM5PwAv67UQC&oi=fnd&pg=PR7&dq=barbier+scarcity+and+frontiers&ots=0348qQtLs3&sig=84JauI7KJDak_7rfaHJHwuXEo2w}
}

@book{smith1982,
  title = {The {{Wealth}} of {{Nations}}: {{Books}} 1-3},
  shorttitle = {The {{Wealth}} of {{Nations}}},
  author = {Smith, Adam},
  year = {1982},
  month = mar,
  edition = {First Edition},
  publisher = {{Penguin Classics}},
  address = {{London}},
  isbn = {978-0-14-043208-4},
  langid = {english}
}

@article{acemoglu2005,
  title = {Institutions as a Fundamental Cause of Long-Run Growth},
  author = {Acemoglu, Daron and Johnson, Simon and Robinson, James A.},
  date = {2005},
  journaltitle = {Handbook of Economic Growth},
  volume = {1},
  pages = {385--472},
  publisher = {{Elsevier}}
}

@article{acemoglu2013,
  title = {Aggregate Comparative Statics},
  author = {Acemoglu, Daron and Jensen, Martin Kaae},
  date = {2013},
  journaltitle = {Games and Economic Behavior},
  shortjournal = {Games and Economic Behavior},
  volume = {81},
  pages = {27--49},
  url = {http://www.sciencedirect.com/science/article/pii/S0899825613000523},
  urldate = {2020-11-24},
  langid = {english}
}

@article{alchian1973,
  title = {The Property {{Right Paradigm}}},
  author = {Alchian, Armen A. and Demsetz, Harold},
  date = {1973},
  journaltitle = {The Journal of Economic History},
  volume = {33},
  number = {1},
  pages = {16--27},
  url = {https://www.cambridge.org/core/journals/journal-of-economic-history/article/the-property-right-paradigm/40FCB24EADBCF3C38FC3429E6C6F96B2}
}

@article{aldenwily2018,
  title = {Collective {{Land Ownership}} in the 21st {{Century}}: {{Overview}} of {{Global Trends}}},
  shorttitle = {Collective {{Land Ownership}} in the 21st {{Century}}},
  author = {Alden Wily, Liz},
  date = {2018-06},
  journaltitle = {Land},
  volume = {7},
  number = {2},
  pages = {68},
  publisher = {{Multidisciplinary Digital Publishing Institute}},
  url = {https://www.mdpi.com/2073-445X/7/2/68},
  urldate = {2020-09-14},
  langid = {english}
}

@article{allen1982,
  title = {The Efficiency and Distributional Consequences of Eighteenth Century Enclosures},
  author = {Allen, Robert C.},
  date = {1982},
  journaltitle = {The Economic Journal},
  volume = {92},
  number = {368},
  pages = {937--953}
}

@article{allen1992,
  title = {Enclosure and the Yeoman: The Agricultural Development of the {{South Midlands}} 1450-1850},
  shorttitle = {Enclosure and the Yeoman},
  author = {Allen, Robert C.},
  date = {1992},
  journaltitle = {Oxford University Press}
}

@article{alston2012,
  ids = {alston2012a},
  title = {The Development of Property Rights on Frontiers: Endowments, Norms, and Politics},
  shorttitle = {The Development of Property Rights on Frontiers},
  author = {Alston, Lee J. and Harris, Edwyna and Mueller, Bernardo},
  date = {2012},
  journaltitle = {The Journal of Economic History},
  volume = {72},
  number = {3},
  pages = {741--770},
  publisher = {{Cambridge University Press}}
}

@article{anderson1990,
  title = {The {{Race}} for {{Property Rights}}},
  author = {Anderson, Terry L. and Hill, Peter J.},
  date = {1990},
  journaltitle = {The Journal of Law \& Economics},
  volume = {33},
  number = {1},
  pages = {177--197},
  publisher = {{[University of Chicago Press, Booth School of Business, University of Chicago, University of Chicago Law School]}},
  issn = {0022-2186},
  url = {http://www.jstor.org/stable/725514},
  urldate = {2020-09-25}
}

@book{aston1987,
  title = {The {{Brenner Debate}}: {{Agrarian Class Structure}} and {{Economic Development}} in {{Pre-industrial Europe}}},
  shorttitle = {The {{Brenner Debate}}},
  author = {Aston, T. H. and Philpin, C. H. E.},
  date = {1987},
  publisher = {{Cambridge University Press}},
  isbn = {0-521-34933-8}
}

@article{baker2003,
  title = {An Equilibrium Conflict Model of Land Tenure in Hunter-Gatherer Societies},
  author = {Baker, Matthew J.},
  date = {2003},
  journaltitle = {Journal of Political Economy},
  volume = {111},
  number = {1},
  pages = {124--173},
  url = {http://www.journals.uchicago.edu/doi/abs/10.1086/344800}
}

@book{baland1996,
  title = {Halting {{Degradation}} of {{Natural Resources}}: {{Is There}} a {{Role}} for {{Rural Communities}}?},
  shorttitle = {Halting {{Degradation}} of {{Natural Resources}}},
  author = {Baland, Jean-Marie and Platteau, Jean-Philippe},
  date = {1996},
  publisher = {{Oxford University Press}},
  langid = {english},
  pagetotal = {444}
}

@book{berry1993,
  title = {No Condition Is Permanent: {{The}} Social Dynamics of Agrarian Change in Sub-{{Saharan Africa}}},
  shorttitle = {No Condition Is Permanent},
  author = {Berry, Sara},
  date = {1993},
  publisher = {{University of Wisconsin Pres}}
}

@incollection{besley2010,
  title = {Property Rights and Economic Development},
  booktitle = {Handbook of {{Development Economics}}, {{Volume}} 5},
  author = {Besley, Timothy J. and Ghatak, Maitreesh},
  editor = {Rodrik, Dani and Rosenzweig, Mark R.},
  date = {2010},
  publisher = {{Elsevier}}
}

@article{binswanger1993,
  title = {South {{African}} Land Policy: {{The}} Legacy of History and Current Options},
  shorttitle = {South {{African}} Land Policy},
  author = {Binswanger, Hans P. and Deininger, Klaus},
  date = {1993},
  journaltitle = {World Development},
  volume = {21},
  number = {9},
  pages = {1451--1475},
  url = {http://www.sciencedirect.com/science/article/pii/0305750X9390127U},
  urldate = {2017-05-22}
}

@article{binswanger1995,
  title = {Power, Distortions, Revolt and Reform in Agricultural Land Relations},
  author = {Binswanger, Hans P. and Deininger, Klaus and Feder, Gershon},
  date = {1995},
  journaltitle = {Handbook of Development Economics},
  volume = {3},
  pages = {2659--2772},
  url = {http://www.sciencedirect.com/science/article/pii/S1573447195300198},
  urldate = {2017-02-02}
}

@book{blaufarb2016,
  title = {The Great Demarcation: The {{French Revolution}} and the Invention of Modern Property},
  shorttitle = {The Great Demarcation},
  author = {Blaufarb, Rafe},
  date = {2016},
  publisher = {{Oxford University Press}}
}

@book{boone2014,
  title = {Property and Political Order in {{Africa}}: {{Land}} Rights and the Structure of Politics},
  shorttitle = {Property and Political Order in {{Africa}}},
  author = {Boone, Catherine},
  date = {2014},
  publisher = {{Cambridge University Press}}
}

@article{borrasjr2012,
  title = {Global Land Grabbing and Trajectories of Agrarian Change: {{A}} Preliminary Analysis},
  shorttitle = {Global Land Grabbing and Trajectories of Agrarian Change},
  author = {Borras Jr, Saturnino M. and Franco, Jennifer C.},
  date = {2012},
  journaltitle = {Journal of agrarian change},
  volume = {12},
  number = {1},
  pages = {34--59},
  publisher = {{Wiley Online Library}}
}

@book{boserup1965,
  title = {The Conditions of Agricultural Growth. {{The}} Economics of Agrarian Change under Population Pressure.},
  author = {Boserup, Ester},
  date = {1965},
  publisher = {{Allan and Urwin}},
  location = {{London}}
}

@article{brenner1976,
  title = {Agrarian Class Structure and Economic Development in Pre-Industrial {{Europe}}},
  author = {Brenner, Robert},
  date = {1976},
  journaltitle = {Past \& present},
  volume = {70},
  number = {1},
  pages = {30--75}
}

@article{bromley1992,
  title = {The Commons, Property, and Common-Property Regimes},
  author = {Bromley, Daniel W.},
  date = {1992},
  journaltitle = {Making the Commons Work},
  pages = {3--16},
  publisher = {{Institute for Contemporary Studies San Francisco, Calif.}}
}

@article{buchanan2000,
  title = {Symmetric Tragedies: {{Commons}} and Anticommons},
  shorttitle = {Symmetric Tragedies},
  author = {Buchanan, James M. and Yoon, Yong J.},
  date = {2000},
  journaltitle = {The Journal of Law and Economics},
  volume = {43},
  number = {1},
  pages = {1--14},
  publisher = {{The University of Chicago Press}}
}

@book{byres1996,
  title = {Capitalism from above and Capitalism from below: {{An}} Essay in Comparative Political Economy},
  author = {Byres, Terence J.},
  date = {1996},
  publisher = {{St. Martin's Press}},
  location = {{New York \& London}},
  isbn = {0-312-16241-3}
}

@article{carlsson1993,
  title = {Global {{Games}} and {{Equilibrium Selection}}},
  author = {Carlsson, Hans and {van Damme}, Eric},
  date = {1993},
  journaltitle = {Econometrica},
  volume = {61},
  number = {5},
  pages = {989}
}

@article{chambers1953,
  title = {Enclosure and Labour Supply in the Industrial Revolution},
  author = {Chambers, Jonathan David},
  date = {1953},
  journaltitle = {The Economic History Review},
  volume = {5},
  number = {3},
  pages = {319--343},
  publisher = {{JSTOR}}
}

@article{chen2017a,
  title = {Untitled {{Land}}, {{Occupational Choice}}, and {{Agricultural Productivity}}},
  author = {Chen, Chaoran},
  date = {2017},
  journaltitle = {American Economic Journal: Macroeconomics},
  volume = {9},
  number = {4},
  pages = {91--121},
  url = {https://www.aeaweb.org/articles?id=10.1257/mac.20140171}
}

@article{chen2023,
  title = {Land {{Misallocation}} and {{Productivity}}},
  author = {Chen, Chaoran and Restuccia, Diego and Santaeulàlia-Llopis, Raül},
  date = {2023},
  journaltitle = {American Economic Journal: Macroeconomics},
  volume = {15},
  number = {2},
  pages = {441--465},
  issn = {1945-7707},
  doi = {10.1257/mac.20170229},
  url = {https://www.aeaweb.org/articles?id=10.1257/mac.20170229},
  urldate = {2023-06-14},
  langid = {english}
}

@article{clark1998,
  title = {Commons {{Sense}}: {{Common Property Rights}}, {{Efficiency}}, and {{Institutional Change}}},
  author = {Clark, Gregory},
  date = {1998},
  journaltitle = {The Journal of Economic History},
  volume = {58},
  number = {1},
  pages = {73--102}
}

@article{coase1960,
  title = {The {{Problem}} of {{Social Cost}}},
  author = {Coase, R. H.},
  date = {1960},
  journaltitle = {The Journal of Law and Economics},
  volume = {3},
  pages = {1--44},
  url = {http://www.journals.uchicago.edu/doi/pdfplus/10.1086/466560}
}

@article{cohen1975,
  title = {A {{Marxian}} Model of Enclosures},
  author = {Cohen, Jon and Weitzman, Martin},
  date = {1975},
  journaltitle = {Journal of Development Economics},
  volume = {1},
  number = {4},
  pages = {287--336}
}

@incollection{conning2007,
  title = {Freedom, {{Servitude}}, and {{Voluntary Contracts}}},
  booktitle = {Buying Freedom: The Ethics and Economics of Slave Redemption},
  author = {Conning, Jonathan and Kevane, Michael},
  editor = {Appiah, Anthony and Bunzl, Martin},
  date = {2007},
  publisher = {{Princeton University Press}}
}

@article{corchon2021,
  title = {Aggregative Games},
  author = {Corchón, Luis C.},
  date = {2021},
  journaltitle = {SERIEs},
  shortjournal = {SERIEs},
  volume = {12},
  number = {1},
  pages = {49--71},
  issn = {1869-4195},
  doi = {10.1007/s13209-021-00229-5},
  url = {https://doi.org/10.1007/s13209-021-00229-5},
  urldate = {2021-10-06},
  langid = {english}
}

@incollection{crafts2004,
  title = {Precocious {{British}} Industrialization: A General Equilibrium Perspective},
  shorttitle = {Precocious {{British}} Industrialization},
  booktitle = {Exceptionalism and {{Industrialisation}}: {{Britain}} and Its {{European Rivals}}, 1688–1815},
  author = {Crafts, N. F. R. and Harley, C. Knick},
  editor = {Prados de la Escosura, Leandro},
  date = {2004},
  pages = {86--107}
}

@article{dasgupta1980,
  title = {Uncertainty, Industrial Structure, and the Speed of {{R}}\&{{D}}},
  author = {Dasgupta, Partha and Stiglitz, Joseph},
  date = {1980},
  journaltitle = {The Bell Journal of Economics},
  pages = {1--28},
  publisher = {{JSTOR}}
}

@book{deininger2003,
  title = {Land {{Policies}} for {{Growth}} and {{Poverty Reduction}}},
  author = {Deininger, Klaus and World Bank},
  date = {2003},
  publisher = {{World Bank Publications}},
  isbn = {978-0-8213-5071-3},
  langid = {english},
  pagetotal = {290}
}

@book{dejanvry1981,
  title = {The {{Agrarian Question}} and {{Reformism}} in {{Latin America}}},
  author = {{de Janvry}, Alain},
  date = {1981},
  publisher = {{ohns Hopkins University Press}},
  location = {{Baltimore}}
}

@article{dejanvry2015a,
  ids = {dejanvry2015b},
  title = {Delinking Land Rights from Land Use: {{Certification}} and Migration in {{Mexico}}},
  shorttitle = {Delinking Land Rights from Land Use},
  author = {{de Janvry}, Alain and Emerick, Kyle and Gonzalez-Navarro, Marco and Sadoulet, Elisabeth},
  date = {2015},
  journaltitle = {American Economic Review},
  volume = {105},
  number = {10},
  pages = {3125--49}
}

@article{demeza1992,
  title = {The {{Social Efficiency}} of {{Private Decisions}} to {{Enforce Property Rights}}},
  author = {{de Meza} and Gould},
  date = {1992},
  journaltitle = {Journal of Political Economy},
  volume = {100},
  number = {3},
  pages = {561--580},
  url = {http://www.jstor.org/stable/2138731}
}

@book{desoto2000,
  title = {The {{Mystery}} of {{Capital}}: {{Why Capitalism Succeeds}} in the {{West}} and Fails Almost Everywhere Else},
  shorttitle = {The {{Mystery}} of {{Capital}}: {{Why Capitalism Succeeds}} in the {{West}} and Fails Almost Everywhere Else},
  author = {{de Soto}},
  date = {2000},
  publisher = {{Basic Books}},
  location = {{New York}}
}

@article{goldstein2008,
  title = {The {{Profits}} of {{Power}}: {{Land Rights}} and {{Agricultural Investment}} in {{Ghana}}},
  shorttitle = {The {{Profits}} of {{Power}}},
  author = {Goldstein, Markus and Udry, Christopher},
  date = {2008},
  journaltitle = {Journal of Political Economy},
  volume = {116},
  number = {6},
  pages = {981--1022},
  issn = {0022-3808},
  urldate = {2017-12-08},

}

@article{gottlieb2019,
  title = {Communal Land and Agricultural Productivity},
  author = {Gottlieb, Charles and Grobovšek, Jan},
  date = {2019},
  journaltitle = {Journal of Development Economics},
  shortjournal = {Journal of Development Economics},
  volume = {138},
  pages = {135--152},
  url = {http://www.sciencedirect.com/science/article/pii/S0304387818304462},
  urldate = {2020-11-28},
  langid = {english}
}

@article{greer2018, 
  author = {Greer, Allan},  
  title = {Property and Disposession: Natives, Empires and Land in Early Modern North America}, 
  year = {2018},
  address = {Cambridge, New York, Port Melbourne},
  publisher = {Cambridge University Press}
}

@book{hatcher2001,
  title = {Modelling the Middle Ages: The History and Theory of {{England}}'s Economic Development},
  shorttitle = {Modelling the Middle Ages},
  author = {Hatcher, John and Bailey, Mark},
  date = {2001},
  publisher = {{OUP Oxford}}
}

@article{heldring2024,
  title = {The Economic Effects of the {{English Parliamentary}} Enclosures},
  author = {Heldring, Leander and Robinson, James A and Vollmer, Sebastian},
  date = {2024},
  journaltitle = {NBER working paper},
  pages = {48},
  langid = {english}
}

@book{holden2013,
  title = {Land {{Tenure Reform}} in {{Asia}} and {{Africa}}: {{Assessing Impacts}} on {{Poverty}} and {{Natural Resource Management}}},
  shorttitle = {Land {{Tenure Reform}} in {{Asia}} and {{Africa}}},
  author = {Holden, Stein and Otsuka, Keijiro and Deininger, Klaus},
  date = {2013},
  publisher = {{Springer}}
}

@article{humphries1990,
  title = {Enclosures, {{Common Rights}}, and {{Women}}: {{The Proletarianization}} of {{Families}} in the {{Late Eighteenth}} and {{Early Nineteenth Centuries}}},
  shorttitle = {Enclosures, {{Common Rights}}, and {{Women}}},
  author = {Humphries, Jane},
  date = {1990},
  journaltitle = {The Journal of Economic History},
  volume = {50},
  number = {1},
  eprint = {2123436},
  eprinttype = {jstor},
  pages = {17--42},
  publisher = {{Cambridge University Press}},
  issn = {0022-0507},
  url = {http://www.jstor.org/stable/2123436},
  urldate = {2022-07-09}
}

@article{kopsidis2015,
  title = {Where Is the Backward {{Russian}} Peasant? {{Evidence}} against the Superiority of Private Farming, 1883–1913},
  shorttitle = {Where Is the Backward {{Russian}} Peasant?},
  author = {Kopsidis, Michael and Bruisch, Katja and Bromley, Daniel W.},
  date = {2015-03-04},
  journaltitle = {The Journal of Peasant Studies},
  volume = {42},
  number = {2},
  pages = {425--447},
  publisher = {{Routledge}},
  issn = {0306-6150},
  doi = {10.1080/03066150.2014.990895},
  url = {https://doi.org/10.1080/03066150.2014.990895},
  urldate = {2021-03-14}
}

@article{legrand1984,
  title = {Labor Acquisition and Social Conflict on the {{Colombian}} Frontier, 1850–1936},
  author = {LeGrand, Catherine},
  date = {1984},
  journaltitle = {Journal of Latin American Studies},
  volume = {16},
  number = {1},
  pages = {27--49}
}

@book{libecap1993,
  title = {Contracting for Property Rights},
  author = {Libecap, Gary D.},
  date = {1993},
  publisher = {{Cambridge university press}},
  url = {https://books.google.com/books?hl=en&lr=&id=F4LRUB361n0C&oi=fnd&pg=PR9&dq=Libecap+contracting+for+property+rights&ots=6v6Mfzj04U&sig=Tg-aJZV3EWFK8G_wckTwM3IQ01c}
}

@book{marx1992,
  title = {Capital: {{Volume}} 1: {{A Critique}} of {{Political Economy}}},
  author = {Marx, Karl},
  date = {1992},
  publisher = {{Penguin Classics (originally published 1862)}}
}

@incollection{morris2003,
  title = {Global {{Games}}: {{Theory}} and {{Applications}}},
  shorttitle = {Global {{Games}}},
  booktitle = {Advances in {{Economics}} and {{Econometrics}}: {{Proceedins}} of the {{Eight World Congress}} of the {{Econometric Society}}},
  author = {Morris, Stephen and Shin, Hyun Song},
  date = {2003},
  pages = {56},
  publisher = {{Cambridge University Press}}
}

@book{murtazashvili2013,
  title = {The Political Economy of the {{American}} Frontier},
  author = {Murtazashvili, Ilia},
  date = {2013},
  publisher = {{Cambridge University Press}}
}

@article{murtazashvili2016a,
  title = {When Does the Emergence of a Stationary Bandit Lead to Property Insecurity?},
  author = {Murtazashvili, Ilia and Murtazashvili, Jennifer},
  date = {2016},
  journaltitle = {Rationality and Society},
  volume = {28},
  number = {3},
  pages = {335--360}
}

@book{north1973,
  title = {The Rise of the Western World: {{A}} New Economic History},
  shorttitle = {The Rise of the Western World},
  author = {North, Douglass C. and Thomas, Robert Paul},
  date = {1973},
  publisher = {{Cambridge University Press}}
}

@book{north1990,
  title = {Institutions, Institutional Change, and Economic Performance},
  author = {North, Douglass Cecil},
  date = {1990},
  publisher = {{Cambridge University Press}},
  location = {{Cambridge ; New York}},
  pagetotal = {viii, 152}
}

@book{onoma2009,
  title = {The Politics of Property Rights Institutions in {{Africa}}},
  author = {Onoma, Ato Kwamena},
  date = {2009},
  publisher = {{Cambridge University Press}}
}

@book{ostrom1990,
  title = {Governing the Commons: {{The}} Evolution of Institutions for Collective Action},
  shorttitle = {Governing the {{Commons}}},
  author = {Ostrom, Elinor},
  date = {1990},
  publisher = {{Cambridge University Press}}
}

@report{otsuka2014,
  title = {Changes in Land Tenure and Agricultural Intensification in Sub-{{Saharan Africa}}},
  author = {Otsuka, Keijiro and Place, Frank},
  date = {2014},
  institution = {{WIDER Working Paper}}
}

@book{pearce2016,
  title = {Common {{Ground}}: {{Securing}} Land Rights and Safeguarding the Earth},
  shorttitle = {Common {{Ground}}},
  author = {Pearce, Fred},
  date = {2016},
  publisher = {{Oxfam International}},
  url = {https://oxfamilibrary.openrepository.com/handle/10546/600459},
  urldate = {2023-10-05}
}

@book{perelman2000,
  title = {The {{Invention}} of {{Capitalism}}: {{Classical Political Economy}} and the {{Secret History}} of {{Primitive Accumulation}}},
  shorttitle = {The {{Invention}} of {{Capitalism}}},
  author = {Perelman, Michael},
  date = {2000},
  publisher = {{Duke University Press}},
  location = {{Durham, NC}},
  isbn = {978-0-8223-2491-1},
  pagetotal = {424}
}

@article{platteau1996,
  ids = {platteau1996a},
  title = {The Evolutionary Theory of Land Rights as Applied to Sub-{{Saharan Africa}}: A Critical Assessment},
  shorttitle = {The Evolutionary Theory of Land Rights as Applied to Sub-{{Saharan Africa}}},
  author = {Platteau, Jean-Philippe},
  date = {1996},
  journaltitle = {Development and Change},
  volume = {27},
  number = {1},
  pages = {29--86}
}

@book{richardcornes1996,
  title = {The Theory of Externalities, Public Goods, and Club Goods},
  author = {Richard Cornes and Todd Sandler},
  date = {1996},
  publisher = {{Cambridge University Press}}
}

@article{samuelson1974,
  title = {Is the {{Rent-Collector Worthy}} of {{His Full Hire}}?},
  author = {Samuelson, Paul A.},
  date = {1974},
  journaltitle = {Eastern Economic Journal},
  volume = {1},
  number = {1},
  pages = {7--10},
  issn = {0094-5056}
}

@book{scott1998,
  title = {Seeing like a State: {{How}} Certain Schemes to Improve the Human Condition Have Failed},
  shorttitle = {Seeing like a State},
  author = {Scott, James C.},
  date = {1998},
  publisher = {{Yale University Press}}
}

@article{shipton1992,
  title = {Introduction. {{Understanding African}} Land-Holding: {{Power}}, Wealth, and Meaning},
  shorttitle = {Introduction. {{Understanding African}} Land-Holding},
  author = {Shipton, Parker and Goheen, Mitzi},
  date = {1992},
  journaltitle = {Africa},
  volume = {62},
  number = {3},
  pages = {307--325},
  url = {https://www.cambridge.org/core/journals/africa/article/introduction-understanding-african-land-holding-power-wealth-and-meaning/ACF66D5AB03387E34DBA3D08DF7D2136}
}

@article{solberg1969,
  ids = {solberg1969a},
  title = {A {{Discriminatory Frontier Land Policy}}: {{Chile}}, 1870-1914},
  shorttitle = {A {{Discriminatory Frontier Land Policy}}},
  author = {Solberg, Carl E},
  date = {1969-10},
  journaltitle = {The Americas},
  volume = {26},
  number = {2},
  pages = {115--133},
  issn = {00031615},
  url = {http://links.jstor.org/sici?sici=0003-1615%28196910%2926%3A2%3C115%3AADFLPC%3E2.0.CO%3B2-Q},
  urldate = {2007-09-25}
}

@misc{takeuchi2022,
  title = {African {{Land Reform Under Economic Liberalisation}}: {{States}}, {{Chiefs}}, and {{Rural Communities}}},
  shorttitle = {African {{Land Reform Under Economic Liberalisation}}},
  author = {Takeuchi, Shinichi},
  date = {2022},
  organization = {{Springer Nature}}
}

@book{tawney1912,
  title = {The Agrarian Problem in the Sixteenth Century},
  author = {Tawney, Richard Henry},
  date = {1912},
  publisher = {{Longmans, Green and Company}}
}

@article{vives2005,
  title = {Complementarities and Games: {{New}} Developments},
  shorttitle = {Complementarities and Games},
  author = {Vives, Xavier},
  date = {2005},
  journaltitle = {Journal of Economic Literature},
  volume = {43},
  number = {2},
  pages = {437--479}
}

@book{wakefield1849,
  title = {A {{View}} of the {{Art}} of {{Colonization}}: {{With Present Reference}} to the {{British}}},
  shorttitle = {A {{View}} of the {{Art}} of {{Colonization}}},
  author = {Wakefield, Edward Gibbon},
  date = {1849},
  publisher = {{J. W. Parker}},
  pagetotal = {513}
}

@article{weitzman1974,
  title = {Free Access vs Private Ownership as Alternative Systems for Managing Common Property},
  author = {Weitzman, Martin L},
  date = {1974},
  journaltitle = {Journal of Economic Theory},
  shortjournal = {Journal of Economic Theory},
  volume = {8},
  number = {2},
  pages = {225--234},
  issn = {0022-0531},
  doi = {10.1016/0022-0531(74)90015-5},
  url = {http://www.sciencedirect.com/science/article/pii/0022053174900155}
}

@book{whittle2013,
  title = {Landlords and {{Tenants}} in {{Britain}}, 1440-1660: {{Tawney}}'s {{Agrarian Problem Revisited}}},
  shorttitle = {Landlords and {{Tenants}} in {{Britain}}, 1440-1660},
  author = {Whittle, Jane},
  date = {2013},
  volume = {1},
  publisher = {{Boydell \& Brewer Ltd.}}
}

@book{wood2002,
  title = {The Origin of Capitalism: {{A}} Longer View},
  shorttitle = {The Origin of Capitalism},
  author = {Wood, Ellen Meiksins},
  date = {2002},
  publisher = {{Verso}}
}

@book{bardhan1989,
    address = {Oxford; New York; Toronto and Melbourne},
    title = {The economic theory of agrarian institutions},
    shorttitle = {The economic theory of agrarian institutions},
    publisher = {Oxford University Press},
    author = {Bardhan, Pranab ed},
    year = {1989},
}

\end{document}